\documentclass[a4paper,onecolumn,11pt,accepted=2026-05-20]{quantumarticle}
\pdfoutput=1
\usepackage[utf8]{inputenc}
\usepackage[english]{babel}
\usepackage[T1]{fontenc}
\usepackage{amsmath}
\usepackage{amsthm}
\usepackage{hyperref}
\usepackage{lipsum}
\usepackage[dvipsnames]{xcolor}
\usepackage{graphicx}
\usepackage{amsmath}
\usepackage{amssymb}
\usepackage{booktabs}
\usepackage{hyperref}
\usepackage{algorithm}
\usepackage{algcompatible}
\usepackage{stmaryrd}
\usepackage{mathtools}
\usepackage{enumitem}
\usepackage{cleveref}

\usepackage{newfloat}
\DeclareFloatingEnvironment[name=Protocol]{protocol}

\DeclareMathSymbol{\qm}{\mathalpha}{operators}{"3F}
\DeclareMathAlphabet{\mathbbold}{U}{bbold}{m}{n}

\newcommand{\ket}[1]{\left |#1 \right \rangle}
\newcommand{\bra}[1]{\left \langle #1 \right|}

\newtheorem{theorem}{Theorem}[section]

\newtheorem{corollary}[theorem]{Corollary}
\newtheorem{lemma}[theorem]{Lemma}

\newtheorem{definition}[theorem]{Definition}

\numberwithin{equation}{section}
\numberwithin{equation}{section}

\newcommand{\HH}{\mathcal{H}}

\newcommand{\id}{\mathbbm{I}}
\newcommand{\CC}{\mathbb{C}}

\newcommand{\NN}{\mathbb{N}}
\newcommand{\RR}{\mathbb{R}}

\newcommand{\BB}{\mathcal{B}}

\newcommand{\XX}{\mathcal{X}}
\newcommand{\YY}{\mathcal{Y}}
\newcommand{\Ss}{\mathcal{S}}
\newcommand{\QQ}{\mathcal{Q}}
\newcommand{\probP}{\text{I\kern-0.15em P}}

\newcommand{\RE}{\mathcal{R}}
\newcommand{\CO}{\mathcal{C}}
\newcommand{\DD}{\mathcal{D}}

\newcommand\abs[1]{\left|#1\right|}
\newcommand{\norm}[1]{\left\lVert#1\right\rVert}

\newcommand{\FUNC}{\mathcal{F}}

\newcommand{\bea}{\begin{eqnarray}} % Shortcut for equation arrays
\newcommand{\eea}{\end{eqnarray}}
\begin{document}

\title{A Practical Protocol for Quantum Oblivious Transfer from One-Way Functions}

\author{Eleni Diamanti}
\affiliation{Sorbonne Université, CNRS, LIP6, 4 Place Jussieu, Paris F-75005, France}
\orcid{0000-0003-1795-5711}

\author{Alex B. Grilo}
\affiliation{Sorbonne Université, CNRS, LIP6, 4 Place Jussieu, Paris F-75005, France}
\orcid{0000-0001-7374-7082}

\author{Adriano Innocenzi}
\affiliation{Sorbonne Université, CNRS, LIP6, 4 Place Jussieu, Paris F-75005, France}
\orcid{0009-0007-0167-8170}

\author{Pascal Lefebvre}
\affiliation{Sorbonne Université, CNRS, LIP6, 4 Place Jussieu, Paris F-75005, France}
\affiliation{KTH Royal Institute of Technology, Stockholm, Sweden}
\orcid{0000-0001-8661-1440}

\author{Verena Yacoub}
\affiliation{Sorbonne Université, CNRS, LIP6, 4 Place Jussieu, Paris F-75005, France}
\orcid{0009-0006-4781-6279}

\author{Álvaro Yángüez}
\affiliation{Sorbonne Université, CNRS, LIP6, 4 Place Jussieu, Paris F-75005, France}
\email{alvaro.yanguez@lip6.fr}
\orcid{0009-0004-7713-6560}

\maketitle
%\

\begin{abstract}
We present a new simulation-secure quantum oblivious transfer (QOT) protocol based on one-way functions in the plain model.
With a focus on practical implementation, our protocol surpasses prior works in efficiency, promising feasible experimental realization.  We address potential experimental errors and their correction, offering analytical expressions to facilitate the analysis of the required quantum resources. Technically, we achieve simulation security for QOT through an \textit{equivocal} and \textit{relaxed-extractable} quantum bit commitment.  

\end{abstract}

\clearpage
\section{Introduction}
\label{sec:introduction}

Since the early stages of quantum cryptography, quantum protocols that achieve information-theoretical security (i.e., security against unbounded adversaries) have been known for primitives that are impossible in the classical setting~\cite{BB84}. However, it was soon realized that even quantumly, a very limited family of cryptographic primitives can achieve such level of security~\cite{LoChau97,Mayers97}, and therefore, the use of computational assumptions is necessary to reach quantum advantage. 

Multi-party computation (MPC) is a very versatile primitive that plays a central role in classical cryptography. In MPC, many parties want to collectively compute a function that depends on their private inputs, while maintaining the inputs secret (even if many of these parties are malicious and deviate from the original protocol). This functionality can be constructed from another primitive: oblivious transfer (OT)~\cite{DFL+09}. However, it is not expected that OT (and thus MPC) can be constructed only from one-way functions (OWFs)\footnote{One-way functions are functions that are easy to compute and hard to invert. These are considered the minimal computational assumption in classical cryptography.}\footnote{More formally, \cite{IR90} shows that no MPC protocol can be built from OWF in a black-box way.}. 

On the other hand, \cite{GLSV21,BCKM21} showed that OT {\em can} be built from OWFs and quantum resources. More precisely, \cite{GLSV21,BCKM21} showed how to build equivocal and extractable commitment schemes, which were known to be sufficient to build quantum protocols for OT ~\cite{DFL+09} and MPC~\cite{IPS08,Kil88}~\footnote{We notice that later, it was shown that MPC can be built even from weaker computation assumptions such as pseudo-random 
states~\cite{MY22,AQY22}.}.

One important feature of~\cite{GLSV21,BCKM21} is that the quantum resources required for such protocols are exactly the same as the ones used in some quantum key distribution (QKD) protocols: they only need to prepare, communicate, and measure one-qubit states in conjugate bases. Given the huge progress in the implementation of QKD in very different setups~\cite{ACA23,GBR23,PSG23} \cite{ZMM24}, it would be expected that one could easily implement such quantum protocols with current technology.

Unfortunately, this is not the case. While the protocols proposed by ~\cite{GLSV21, BCKM21} are important to theoretically understand the power of quantum resources, their building blocks put serious barriers in their experimental implementation. More concretely, they pose the following difficulties:
\begin{itemize}
    \item \textbf{Fragility against errors.} The protocols proposed by \cite{GLSV21} and  \cite{BCKM21} do not tolerate experimental errors such as bit flips during the state distribution steps. 
    \item \textbf{Practical hash functions and zero knowledge proofs.} Zero knowledge proofs used in \cite{GLSV21} do not make black-box use of OWFs, and the practical implementation of the protocol would require the arithmetical description of the inner functioning of the one-way functions.  It is currently unclear how to achieve such descriptions for widespread heuristic implementations of post-quantum OWFs like the secure hash standard (SHA), making their integration in the \cite{GLSV21} scheme non-trivial.
    \item \textbf{Inefficiency.} The iterative structure of the protocol proposed by \cite{BCKM21} allows the protocol to make a black-box use of the underlying components but requires a very large quantity of quantum states, which prevents it from being practically implementable. For example, one run of that protocol requires an order of  $10^{13}$ BB84 states~\cite{BB84}, which are provided by days of transmission given a state-of-the-art discrete-variable QKD setup.
\end{itemize}

The main contribution of this work is to provide a noise-tolerant protocol for OT, based on the structure from \cite{BCKM21}, while avoiding some of its bottlenecks and while making it efficiently implementable. In particular, with this new protocol, we expect that around $10^{7}$ BB84 states would be sufficient instead of $10^{13}$, improving the transmission time to the order of seconds. Moreover, we notice that with this amount of quantum resources, we can distill more than one OT at once, which will be crucial in future uses of this protocol to implement MPC.

\subsection{Background and our results}
\label{sec:intro-technical}
As previously mentioned, the goal of \cite{GLSV21,BCKM21} is to achieve a quantum protocol for oblivious transfer, since it can be used to construct MPC in a generic way~\cite{IPS08}. Oblivious transfer is a cryptographic primitive where Alice chooses two messages, $m_0$ and $m_1$,  Bob chooses a bit $b \in \{0,1\}$, and Bob learns $m_b$ (notice that Bob should not learn $m_{\bar{b}}$ and Alice should not learn $b$).

In the quantum protocol for OT proposed in \cite{CK88,BBCS92}, Alice sends BB84 states to Bob. The encoded states are $\left(\ket{x^{OT}_i}_{\theta^{OT}_i}\right)_{i \in [2\lambda_{OT}]}$, where $x^{OT}_i\in\{0,1\}$ is the value of the bit, $\theta^{OT}_i \in \{+, \times\}$ is the choice of encoding bases and $\lambda_{OT}$ is a security parameter. Bob chooses measurement bases $\hat{\theta}^{OT}_i \in\{+,\times\}$ and obtains measurement results $\hat{x}^{OT}_i\in\{0,1\}$. Since we cannot guarantee that a malicious Bob will follow the protocol (in particular, that he keeps these qubits in a quantum memory instead of measuring them), Alice and Bob perform a sub-protocol that will give guarantees to Alice that Bob did measure the quantum state. However, this sub-protocol comes with a cost, since it adds the need for bit commitment, and requires the added cost of cryptographic assumptions. More concretely, to ensure Bob's measurement, Bob {\em commits} to the measurement bases that he chose and the outcomes resulting from the measurement. Alice can then choose a subset of these pairs of measurements/outcomes, and she asks Bob to {\em open} the corresponding commitments. Then, she can check if the outcomes are consistent with the BB84 states that she sent. If they are, it can be shown that there are strong guarantees that Bob measured (most of) the qubits sent by Alice~\cite{BF12}. The bit commitment sub-protocol is finished when Alice sends the bases of her original BB84 states. Bob then divides his indices into two different subsets $I_b=\{i:\theta^{OT}_i=\hat{\theta}^{OT}_i\}$ and $I_{\bar{b}}=\{i:\theta^{OT}_i\neq\hat{\theta}^{OT}_i\}$, which depend on the bit choice bit $b$.  Bob sends the two sets $I_0$ and $I_1$ to Alice who can encode the messages $m_0$ and $m_1$ using $\mathbf{x}_0^{OT}=\{x^{OT}_i: i\in I_0\}$ and $\mathbf{x}_1^{OT}=\{x^{OT}_i: i\in I_1\}$ as encrypting keys. Bob receives the encoded messages and decodes $m_b$.

The subtleties here are the properties needed in the commitments to prove that the protocol is secure. In standard bit commitment schemes, there are two properties of interest: {\em hiding}, meaning that the receiver cannot learn the committed message before the opening, and {\em binding}, meaning that after committing, there is only one value that can be opened by the committer. However, these properties alone do not allow us to prove the security of the quantum OT protocol in the simulation-based setting (which is needed to use this building block in MPC protocols). To prove the simulation-security of the QOT protocol, a strengthening of these properties is needed, namely:

\smallskip

\noindent\textbf{Equivocality:} A bit commitment protocol is called equivocal if there is a simulator (also called equivocator) that can perform a “dummy commitment” that can be opened to any desired value. Moreover, such a dummy commitment cannot be distinguished from the real ones by any polynomial-time distinguisher.~\footnote{We notice that for this to be true, the simulator has to have some leverage (e.g., a trapdoor, or the ability to rewind), otherwise the protocol would not be secure. \label{footnote:simulator}} 
    
\noindent\textbf{Extractability:} A bit commitment protocol is called extractable if there is a simulator (also called extractor) that is able to extract the commitment message in the commitment phase of the protocol.~\footnote{We have here the same considerations as in \Cref{footnote:simulator}.} 

\medskip

In \cite{DFL+09}, they show that equivocation and extraction of the quantum bit commitment are sufficient properties to prove the simulation-security of the QOT protocols of \cite{CK88,BBCS92}.
Moreover, the technical contribution of \cite{GLSV21,BCKM21} was to show a quantum protocol bit commitment that is equivocal and extractable. The main observation that allowed us to improve the parameters of such a protocol is that full extraction is not needed in order to prove the security of the QOT protocol. Instead, we show that {\em relaxed extractability} is sufficient, and we show a more efficient bit commitment protocol that achieves this property (and is still equivocal). We describe the details of our contributions in the next section.

\subsection{Technical overview}
We explain now the approach of \cite{BCKM21} to achieve extractable and equivocal commitment, and discuss how we modify it. In their result, they propose two compilers: in the first one, they propose an equivocal commitment scheme based on any bit commitment scheme that satisfies only “vanilla” binding and hiding. Moreover, the resulting commitment is extractable if the original commitment was also extractable. In the second commitment, they show how to turn an equivocal commitment into an extractable one (while losing equivocality). We notice that in this second step, even if the original bit commitment scheme is classical, the resulting one would be a quantum protocol. They achieve the equivocal and extractable bit commitment scheme by starting with a standard bit commitment scheme, and then applying the equivocal compiler, followed by the extractable compiler, and then finally the equivocal compiler again.

In the equivocal compiler, \textit{EqCommitment}, the committer generates $\lambda$ pairs of random bits $(u_i^0,u_i^1)$. For each $i$, the committer and the receiver proceed as follows, sequentially: the committer commits, using the base bit commitment scheme, to each one of them twice leading to the commitments $(c_{i,0}^0,c_{i,1}^0,c_{i,0}^1,c_{i,1}^1)$. Then, the receiver chooses a random bit $\gamma_i$ and the committer opens $c_{i,0}^{\gamma_i},c_{i,1}^{\gamma_i}$. The verifier aborts if such pairs of commitments open to different values. At the end of this interaction, the committer sends $e_i=b \oplus u^{\bar{\gamma}_i}$ for every $i$. In the opening phase, the committer chooses bits $d_i$ and opens $c_{i,d_i}^{\bar{\gamma}_i}$ and the receiver accepts if all the $e_i$'s decommit to the same value $b$. 

In the equivocality proof, the equivocator guesses $\gamma_i$, and then commits to two different values on $c_{i,0}^{\bar{\gamma}_i},c_{i,1}^{\bar{\gamma}_i}$. At each step, the equivocator uses Watrous' rewinding~\cite{Wat05} to amplify the success probability to $1-negl(\lambda)$. In this case, in the opening phase, the equivocator can choose which $c_{i,d_i}^{\bar{\gamma}_i}$ to open, so that the decommitment will be the desired bit. The binding/equivocal property follows from the fact that, if many of the checks on $\gamma_i$'s passed, then most of the pairs $c_{i,0}^{\bar{\gamma}_i},c_{i,1}^{\bar{\gamma}_i}$ commit to the same value, and therefore the decommitted bit is fixed.

The modification that we propose to the equivocal compiler is simple but impactful. Instead of repeating the commitment/checking $\lambda$ times sequentially, our protocol only performs it once. The proof of equivocality follows exactly as in \cite{BCKM21}, but the binding property completely breaks. In particular, a malicious committer can guess $\gamma$, and commit to two different values of $c_{0}^{\bar{\gamma}},c_{1}^{\bar{\gamma}}$, and therefore, with probability $\frac{1}{2}$, the receiver's checks pass and the committer can open to any value of their choice. However, we show that this protocol satisfies what we call {\em relaxed binding}. We defer the formal definition of this property to \Cref{relaxedstatisticalbinding}, but, intuitively, it says that, if we commit to $m$ bits using our bit commitment scheme and all of the tests pass, then, with overwhelming probability, only a logarithmic small fraction of these bits are non-binding. As we discuss next, we show that this property is sufficient in the extractable compiler. Moreover, since this sub-protocol is repeated many times in later parts of the protocol, reducing its complexity has a big impact in the runtime of the classical post-processing of our final protocol.

\medskip

We switch gears now to the extractable compiler of \cite{BCKM21}. This compiler is heavily inspired by the structure of the quantum OT protocol of \cite{CK88,BBCS92}. In this compiler, the committer  generates $2\lambda$ BB84 states and sends them to the receiver. As in the quantum OT protocol, the receiver equivocally commits to measurement bases and outcomes, then the committer challenges the receiver to open a subset of size $\lambda$ of such commitments and verify the consistency of the opened values with the committer's original encoded states. If such a test passes, the committer divides the remaining encoded states into $\sqrt{\lambda}$ strings $x_1,...,x_{\sqrt\lambda}$ of size $\sqrt{\lambda}$, and using $\sqrt{\lambda}$ hashes from a $2$-universal hash function $h_1,...,h_{\sqrt\lambda}$, the committer sends $\hat{h}_i = h_i(x_i) \oplus b$ 

The opening in \cite{BCKM21} is followed by the committer sending  $x_1,...,x_{\sqrt\lambda},b$, which is followed by the receiver checking if these values are consistent with their measurement outcomes and if $\hat{h}_i = h_i(x_i) \oplus b$.

The hiding property of the commitment comes from the properties of $2$-universal hash functions along with the entropy of $x_i$ from the receiver's perspective due to the uncertainty relations of measurement outcomes that we can achieve with the binding property of the commitment scheme. We notice that the relaxed binding of the base commitment instead of a full binding preserves the overall proof of hiding with minimal losses.

The extraction property is (roughly) proved as follows. Using the equivocality of the base bit commitment scheme, the extractor can delay the measurement until the committer reveals which subset of positions they will check. At this point, the extractor measures those positions on random bases, and then equivocates the opening to these values. Later in the protocol, whenever the committer sends the bases, the extractor is able to measure {\em all} of the qubits in the correct bases and find the (purported) values of $x_1,...,x_{\sqrt\lambda}$, and then extract the committed values from $\hat{h}_i$.

The first problem of this compiler is that it assumes that all of the parties have access to perfect (noiseless) devices. In particular, noise will make it impossible for the extractor to extract the correct values. A second problem appears while splitting the string in $\sqrt{\lambda}$ blocks, since it causes a quadratic loss in the security of the protocol. Finally, we notice that, in \cite{BCKM21}, this protocol is repeated multiple times: firstly in the equivocal compiler, and then to commit to many values in the final quantum OT protocol. This causes the huge overhead needed in  the total number of qubits required for their protocol. 

In our work, we solve all of these problems at once. First, in order to enable extraction even in the presence of noise, we send the syndrome of the encoded values according to a linear error correcting code, and, as in the QKD setting, this allows the extractor to correct the faulty positions. Second, we replace the use of hashes so that, with the same number of BB84 states, we can commit to many qubits with equivocality and some extraction properties. Our proposed Equivocal and Relaxed-Extractable commitment scheme (\textit{ERE-Commitment}), can be divided into two main different subroutines: the generation of a set of random seeds that are relaxed-extractable and the use of these seeds as “keys” in a set of equivocal commitments. We show that the keys of almost all the equivocal commitments are extractable, which is sufficient for a simulator to open most of the messages. We notice that the structure of our scheme is conceptually different from the one proposed by \cite{BCKM21}, allowing the drastic reductions of the needed quantum resources.

The challenge is to generate random seeds that are relaxed-extractable. To do so, we make use of quantum resources again. Given the distribution of BB84 states, the sender can distill random seeds from the encoded states. By using an equivocal commitment subroutine, these seeds cannot be distilled by a receiver that is forced to measure. However, an efficient simulator is able to extract them.

More concretely, we split the encoded string $\boldsymbol{x}_{EX}$  into $\left( \Tilde{\boldsymbol{x}}^j\right)_{j \in [2k]} $ subsets of size $m$. Using privacy amplification techniques found in QKD \cite{Ren05}, we distill uniformly random strings $\boldsymbol{s}^j$ which are used as seeds of pseudo-random generators, i.e., we expand them to the pseudo-random strings $\left(\boldsymbol{p}^j\right)_{j \in [2k]}$. Each $\boldsymbol{p}^j$ represents a concatenation of $2w$ values $\boldsymbol{p}^j_i$ that are used for their randomness as needed in statistically binding bit commitment schemes used for equivocation.

ERE-Commitment allows us to achieve the independence of seeds required to open the different commitments, while avoiding the modular composition proposed by \cite{BCKM21}. The reason is that we consider that the $2k$ seed families $(\boldsymbol{p}^j)_{j \in [2k]}$ are grouped in $k$ pairs $(\boldsymbol{p}^1, \boldsymbol{p}^2), ...(\boldsymbol{p}^{2k-1},\boldsymbol{p}^{2k})$, so each of the $r \in [k]$ pairs can be used for parallel commitment of $w$ values. For the $r$-th pair, each of the $q \in [w]$ instances makes use of two seeds $\boldsymbol{p}^j_i$ from each family of the pair, in such a way that the final set of seeds for the $q$-th EqCommitment instance of the $r$-th family pair is $\left(\left(\boldsymbol{p}_{i_{q,0}}^{j^{r,0}},\boldsymbol{p}_{i_{q,1}}^{j^{r,0}}, \boldsymbol{p}_{i_{q,0}}^{j^{r,1}},\boldsymbol{p}_{i_{q,1}}^{j^{r,1}} \right)_{q \in [w]} \right)_{r \in [k]} $. Afterwards, Bob sends Alice the corresponding $x^{EX}_i$ and $\theta^{EX}_i$, then Alice checks that $x^{EX}_i = \hat{x}^{EX}_i $ whenever $\theta^{EX}_i = \hat{\theta}^{EX}_i$ and that the $\boldsymbol{p}^j$ were generated in an honest way. This ensures that the seeds $\boldsymbol{p}_i^j$ originate from the proper $\mathbf{x}^{EX}$ with probability $1-negl(k)$. Bob can then commit to $\left((\hat{x}_i^{OT},\hat{\theta}_i^{OT})\right)_{i \in [2\lambda_{OT}]}$ using $\boldsymbol{p}_i^j$.

To prove that ERE-Commitment is relaxed extractable, we make use of equivocality and quantum rewinding.  A simulator commits to dummy values instead of $\left((\hat{x}^{EX}_i,\hat{\theta}^{EX}_i)\right)_{i\in [4 \lambda_{EX}]}$. After learning the challenged subset $E$, the simulator measures the corresponding states $\left(\ket{x^{EX}_i}_{\theta^{EX}_i}\right)_{i \in E}$, and, applying Watrous' quantum rewinding, opens to the measured values.  Once the simulator has passed the decommitment check, Bob announces the bases. The simulator will then measure $\left(\ket{x^{EX}_i}_{\theta^{EX}_i}\right)_{i \in \bar{E}}$ in the correct bases obtaining $\mathbf{x}^{EX}$. It is in this step in which the simulator has to make use of the error correction code for obtaining the correct seeds $(\boldsymbol{p}^j)_{j \in [2k]}$ given the syndromes $(Synd_j)_{j \in [2k]}$. The simulator will obtain the seeds of the equivocal commitments, and thus extracts  $\abs{\bar{T}}-\omega(\log^2(k))$ of the non-challenged committed values $\left((\hat{x}_i^{OT},\hat{\theta}_i^{OT})\right)_{i \in \bar{T}}$ with probability $1-negl(k)$.

\begin{protocol}[t!]
  \label{prot:our-qot}
  \centering
  \begin{tabular}{@{}p{0.42\linewidth} c p{0.42\linewidth}@{}}
    \multicolumn{1}{@{}c}{\textbf{Alice }} & &
    \multicolumn{1}{c@{}}{\textbf{Bob }} \\[0.6em]

    \textit{Input:} Messages $(m_0,m_1)$ & &
    \textit{Input:} Bit $b \in \{0,1\}$ \\[0.2em]
    \multicolumn{3}{@{}l}{\textit{Assumptions:} ERE-Commitment  PRG, and 2-universal hash functions $h_j$.}\\[0.8em]

    % Step 1
    \textbf{1.} Chooses $\boldsymbol{x} \leftarrow \{0,1\}^{2\lambda_{\mathrm{OT}}}$ and
    $\boldsymbol{\theta} \leftarrow \{+,\times\}^{2\lambda_{\mathrm{OT}}}$, prepares
    $\ket{x}_\theta$, and sends the $2\lambda_{\mathrm{OT}}$ BB84 states. &
    $\xrightarrow{\text{Quantum}}$ &
    \textbf{1.a} Receives the BB84 states. \\[0.8em]
    \multicolumn{1}{@{}p{0.42\linewidth}}{} & &
    \textbf{1.b} Samples measurement bases
    $\hat{\boldsymbol{\theta}} \leftarrow \{+,\times\}^{2\lambda_{\mathrm{OT}}}$ and
    measures to obtain $\hat{\boldsymbol{x}} \leftarrow \{0,1\}^{2\lambda_{\mathrm{OT}}}$\\[1em]
    %Step2
    \textbf{2.} Acts as receiver in the ERE-Commitment execution for
    $(\hat x_i,\hat\theta_i)_{i\in[2\lambda_{\mathrm{OT}}]}$. &
    $\xleftrightarrow{\mathsf{ERE\text{-}Commit}}$ &  \textbf{2.} Executes ERE-Commitment to commit to
    $(\hat x_i,\hat\theta_i)_{i\in[2\lambda_{\mathrm{OT}}]}$.
    \\[1em]

    % Step 3
    \textbf{3.} Chooses a random challenge subset
    $T \subset [2\lambda_{\mathrm{OT}}]$ with $|T| = \lambda_{\mathrm{OT}}$
    and sends $T$. &
    $\xrightarrow{\text{Classical}}$ &
    \textbf{3.} Receives $T$. \\[0.8em]

    % Step 4
    \multicolumn{1}{@{}p{0.42\linewidth}}{\textbf{4.} Runs ERE-Decommitment for each $i \in T$ and obtains
    $(\hat x_i,\hat\theta_i)_{i\in T}$.} &
    $\xleftrightarrow{\mathsf{ERE\text{-}Decommit}}$ &
    \multicolumn{1}{p{0.42\linewidth}@{}}{\textbf{4.} Opens the corresponding commitments on all $i \in T$.} \\[1.5em]

    % Step 5
    \textbf{5.} Checks that $x_i = \hat x_i$ whenever $\theta_i = \hat\theta_i$
    on $T$, up to experimental error $\alpha$; if the test fails, aborts. & & 
     \\[1.2em]

    % Step 6
    \textbf{6.} Discards $T$, reorders $(x_i,\theta_i)_{i\in\bar T}$, and sends
    $\theta$ (on $\bar T$). &
    $\xrightarrow{\text{Classical}}$ &
    \textbf{6.} Discards $T$, reorders $(\hat x_i,\hat\theta_i)_{i\in\bar T}$,
    and receives $\boldsymbol{\theta}$. \\[1em]

    % Step 7
    \multicolumn{1}{@{}p{0.42\linewidth}}{\textbf{7.} Receives $(I_0,I_1)$.} &
    $\xleftarrow{\text{Classical}}$ &
    \multicolumn{1}{p{0.42\linewidth}@{}}{\textbf{7.} Partitions $\bar T$ into
    $I_b = \{ i : \theta_i = \hat\theta_i \}$ and
    $I_{\bar b} = \{ i : \theta_i \neq \hat\theta_i \}$; sends $(I_0,I_1)$.} \\[0.8em]

    % Step 8
    \textbf{8.} Defines $\boldsymbol{x}_0 = \{ x_i : i \in I_0 \}$,
    $\boldsymbol{x}_1 = \{ x_i : i \in I_1 \}$ and sends syndromes
    $\mathsf{Synd}_0,\mathsf{Synd}_1$ for correcting a proportion $\alpha$
    of errors in $\boldsymbol{x}_0,\boldsymbol{x}_1$. &
    $\xrightarrow{\text{Classical}}$ &
    \textbf{8.} Receives $\mathsf{Synd}_0,\mathsf{Synd}_1$. \\[0.8em]

    % Step 9
    \textbf{9.} Samples seeds $s_0,s_1$ and sends
      $(s_0, \mathsf{PRG}(h(s_0,\boldsymbol{x}_0)) \oplus m_0)$ and $
       (s_1, \mathsf{PRG}(h(s_1,\boldsymbol{x}_1)) \oplus m_1)$ &
    $\xrightarrow{\text{Classical}}$ &
    \textbf{9.} Receives the seeds and masked messages. \\[0.8em]

    % Step 10
    \multicolumn{1}{@{}p{0.42\linewidth}}{} & &
    \multicolumn{1}{p{0.42\linewidth}@{}}{\textbf{10.} Performs error correction on $\boldsymbol{x}_b$ using $\mathsf{Synd}_b$ and
    decrypt $m_b$ using $\hat{\boldsymbol{x}}_b = \{ \hat x_i : i \in I_b \}$. If error correction
    fails, chooses $\boldsymbol{x}_b$ as a random string in $\{0,1\}^{\lambda_{\mathrm{OT}}}$.} \\

  \end{tabular}
\caption{Our Proposed QOT protocol.}
\label{protocol:QOT}

\end{protocol}

Despite these changes, we prove that the OT functionality remains simulation-based secure given access to an equivocal and $\chi$-relaxed extractable bit commitment. With a practical implementation as our main goal, we explicitly take into account possible experimental errors and their correction, and we derive analytical expressions based on \cite{BF12} that facilitate benchmarking of the required quantum resources. The resulting QOT protocol is given in Protocol~\ref{protocol:QOT}. 

Moreover, we propose a method to distill a number $n_{\mathrm{OT}}$ of QOT keys from a single run of the protocol, which is particularly useful when a large number of BB84 states is employed to achieve a smaller sampling error (and hence a tighter $\Delta$-security bound).

\subsection{Related works}

We compare and benchmark our proposed protocol with other composably secure QOT protocols based on one-way functions that have been proposed in the literature.
In \cite{ABKK23}, the authors also focus on reducing the round complexity, but they work in the Random Oracle Model (ROM). In this setting, equivocality and extractability of commitments can be obtained from the random oracle, so no additional quantum resources are needed for the commitment layer itself. The cost of the quantum communication in \cite{ABKK23} comes entirely from the OT core, which is structurally similar to CK88. Under our common benchmarking assumptions (target trace distance, security level, and an upper bound on the number of RO queries), the resulting number of BB84 states in their protocols is of the same order of magnitude as in our construction, as shown in \Cref{table:benchmark} and detailed in \Cref{app:benchmark}. However, these estimates for \cite{ABKK23} do not take into account practical errors such as bit-flips, photon loss, or multi-photon events, whereas our benchmark explicitly includes such imperfections.

It is important to stress that the comparison in \Cref{table:benchmark} concerns \emph{only} the quantum communication cost (number of BB84 states) under a common concrete security target. In particular, we do not claim to match the classical efficiency of ROM-based commitments. Conceptually, our assumption is strictly weaker: our protocol is proved secure in the plain model under the minimal assumption of post-quantum one-way functions, without relying on a random oracle or heuristic instantiations thereof. In this sense, our construction achieves comparable BB84 cost while working under strictly milder computational assumptions than \cite{ABKK23}.

\begin{table}
    \centering
    % Make rows taller and columns a bit wider
    \renewcommand{\arraystretch}{1.5}   % vertical padding (default is 1.0)
    \setlength{\tabcolsep}{8pt}        % horizontal padding (default is 6pt)

    \begin{tabular}{|c|c|c|c|c|c|}
    \hline
      & \cite{BCKM21} 
      & \multicolumn{2}{|c|}{\cite{ABKK23}} 
      & \multicolumn{2}{|c|}{\textbf{Our protocol}} \\
    \hline \hline
    $\alpha$ 
      & $-$ 
      & $-$ 
      & $-$ 
      & $-$ 
      & $0.006$ \\
    \hline
    $\vartheta$ 
      & $-$ 
      & $-$ 
      & $-$ 
      & $-$ 
      & $0.001$ \\
    \hline
    $q_{\mathrm{RO}}$ 
      & $-$ 
      & $2^{64}$ 
      & $2^{64}$ 
      & $-$ 
      & $-$ \\
    \hline
    $N_{\mathrm{BB84}}$ 
      & $2.27\times 10^{13}  $ 
      &  $3.22\times 10^{6}$  
      &  $1.43\times 10^{6}$
      & $3.33 \times 10^{7} $  
      & $7.47 \times 10^{7}$  \\
    \hline
    $T_{\mathrm{acq}}$ 
      & 8.6 months 
      & $3.22$ s
      & $1.43$ s  
      & $33.3$ s
      & $74.7$ s  \\
    \hline
    \end{tabular}

\vspace{5mm}
    \caption{Performance benchmark for a target trace distance $\Delta \leq 10^{-15}$. Here, $\alpha$ denotes the bit-flip probability, $\vartheta$ the fraction of leaked bits due to, for example, multi-photon events, $q_{\mathrm{RO}}$ the number of queries to the random oracle (when applicable), $N_{\mathrm{BB84}}$ the number of exchanged BB84 states, and $T_{\mathrm{acq}}$ the corresponding acquisition time assuming a rate of $1$ MHz. From left to right, the columns correspond to: the protocol of \cite{BCKM21}, the 3-round and 4-round communication protocols of \cite{ABKK23}, and our protocol in the ideal (error-free) and realistic noisy settings.}
    \label{table:benchmark}
\end{table}

On the other hand, the number of BB84 states necessary to instantiate \cite{BCKM21} has been obtained by adapting our security bounds to their protocol structure. A detailed explanation of how all entries in \Cref{table:benchmark} are computed is provided in \Cref{app:benchmark}.

Further proposals, such as \cite{CCLY23}, consider relaxed version of simulations, where the run-time of the simulator depends  polynomially on the desired security level. In our work, we can achieve negligible distinguishing probability with polynomial-time simulators. 

\subsection{Paper Organisation}

Section \ref{sec:introduction} contains important background information and the technical overview of the paper.
Section \ref{sec:preldef} contains preliminary information such as a description of the notation used throughout the paper, and useful definitions.
Section \ref{sec:eqcomm} contains the proposed equivocal commitment.
Section \ref{sec:EREcom} presents the equivocal and relaxed-extractable commitment.
Section \ref{sec:QOT} contains the proof that the QOT protocol we propose is simulation-based secure given relaxed-extractable and equivocal bit commitment.

\section{Preliminaries and Definitions} 
\label{sec:preldef}

\subsection{Notation and acronyms}
Throughout this paper, we represent classical random variables with capital letters (e.g., $X$) while the value they take is represented by lower case letters (e.g., $x$).  Bold characters (e.g., $\mathbf{x}$) are strings of values.  Quantum states are represented as density matrices $\rho \in \BB_1(\HH)$. The probability of distinguishing two density matrices is upper-bounded by $\frac{1}{2}(1+\Delta(\rho^{1},\rho^{2}))$, where $\Delta(\rho^{1},\rho^{2})$ is the trace distance. When a negligible function $\mu(\lambda)$ exists such that $\Delta(\rho^{1},\rho^{2}) \leq \mu(\lambda)$, $\rho^{1}$ and $\rho^{2}$ are said to be statistically close. $\mu(\lambda)$ is negligible if, for every fixed $c$, $\mu(\lambda) = o(1/\lambda^c)$, where $\lambda$ is the \textit{security parameter}.  Throughout the paper, the security parameters of the different steps are written as $\lambda_{OT}$ and $\lambda_{EX}$, for example.  When using the symbol $\perp$, it means to abort.

In the Landau notation, given two functions $f(n)$ and $g(n)$, we write $f(n)=o(g(n))$ if $\lim_{n\rightarrow \infty} f(n)/g(n) = 0$. In the same way, $f(n)=\omega(g(n))$ if $\lim_{n\rightarrow \infty} f(n)/g(n) = \infty$ and $f(n)=O(g(n))$ if there exists a constant $ C>0$ such that $\lim_{n\rightarrow \infty} f(n)/g(n) \leq C$. Lastly, $f(n)=\Omega(g(n))$ if there exists a constant $ C>0$ such that $\lim_{n\rightarrow \infty} f(n)/g(n) \geq C$ and $f(n)=\Theta(g(n))$ if both $f(n)=O(g(n))$ and $f(n)=\Omega(g(n))$.

We define a \textit{classical-quantum hybrid state} as a bipartite state of the subsystems $XE$ such that $\rho_{XE} = \sum_{x \in \XX} p_x(x) \ket{x}\bra{x} \otimes \sigma_E^x$, where subsystem $X$ is classical. 
  $\XX$ is a finite set with $\abs{\XX} = dim(\HH_X)$, $p_x: \XX \rightarrow [0,1]$ is a probability distribution and $\{\ket{x}\}_{x \in \XX}$ an orthonormal basis of $\HH_x$.  $\sigma_E^x \in \HH_E$ is the density matrix correlated with the value $x \in \XX$. When we treat a collection of qubits as a single unit, we refer to it as a \textit{quantum register}.

A function $ h : \XX \times \Ss \rightarrow \YY$ is called \textit{2-universal hash} function if, for every two strings $x \neq x' \in \XX$, we have that $\Pr[h(x,S)=h(x',S)] \leq \frac{1}{\abs{\YY}}$. In the same way, a function $G:\{0,1\}^{\ell}\rightarrow\{0,1\}^{n}$ with seed length $\ell < n$  is a \textit{pseudorandom generator} (PRG) if, given a non-uniform quantum polynomial time distinguisher $\DD^* = \{\DD^*_{\lambda}, \sigma_{\lambda}\}$, there exists a negligible function $\nu(\cdot)$ such that $\abs{\probP_{s \in \{0,1\}^n} [\DD^*(1^{\lambda},G(s))=1] - \probP_{r \in \{0,1\}^\ell} [\DD^*(1^{\lambda},r)=1]} \leq \nu(\lambda)$ with $ r$ chosen uniformly at random.

A \textit{non-uniform quantum polynomial time distinguisher}  $\DD^* = \{\DD^*_{\lambda}, \sigma_{\lambda}\}$ is composed of a $\lambda$-size quantum circuit $\DD^*_{\lambda}$ and a non-uniform $\lambda$-size quantum advice $\sigma_{\lambda} \in \BB(\HH_{\DD})$.

For convenience, we summarize the main acronyms used throughout the paper in \Cref{tab:acronyms}.

\begin{table}[ht]
  \centering
  \begin{tabular}{ll}
    \toprule
    Acronym / symbol & Meaning \\
    \midrule
    QOT      & Quantum Oblivious Transfer \\
    OT       & (Classical) Oblivious Transfer \\
    OWF      & One-Way Function \\
    QKD      & Quantum Key Distribution \\
    MPC      & Multi-Party Computation \\
    ROM      & Random Oracle Model \\
    ERE      & Equivocal and Relaxed-Extractable (commitment) \\
   $Com_s$      & Classical Bit Commitment with seed $s$,e.g., Naor's Bit Commitment \\
    BB84     & Standard qubit states $\{\ket{0},\ket{1},\ket{+},\ket{-}\}$ \\
    $\Delta$ & Trace-distance distinguishing advantage \\
    $\lambda$ & Security parameter \\
    $\lambda_{\mathrm{OT}}$, $\lambda_{\mathrm{EX}}$ & Security parameters for OT and extractable layers \\
    $\lambda_{PQS}$ & String length of a seed of a post-quantum safe PRG\\
    $\mathcal{C},\mathcal{R}$    & Committer and receiver in a bit-commitment protocol \\
    $A,B$    & Alice (sender) and Bob (receiver) in the QOT protocol \\
    $ \mathcal{D}^\ast$ & (Non-uniform) quantum polynomial-time distinguisher \\
    PRG      & Pseudorandom generator $G:\{0,1\}^\ell \to \{0,1\}^n$ \\
    $h(\cdot)$      & (Two-)universal hash function $h:X\times S \to Y$ \\
    $\alpha$ & Physical bit-flip probability\\
    $\vartheta$ & Fraction of leaked bits (e.g.\ due to multiphoton events) \\
    $\eta$   & Relaxed binding fraction in $\eta$-relaxed binding \\
    $\chi$   & Relaxed extractability fraction in $\chi$-relaxed extractability \\
    \bottomrule
  \end{tabular}
  \vspace{4mm}
  \caption{Acronyms and main symbols used throughout the paper.}
  \label{tab:acronyms}
\end{table}

\subsection{Bit commitment} 

A bit commitment is a two-party interactive functionality between a quantum committer $\CO$ and a quantum receiver $\RE$. A commitment protocol is composed of two phases: the committing phase and the decommitting phase, where $\CO = \{\CO_{com}, \CO_{dec}\}$ and $\RE = \{\RE_{com}, \RE_{dec}\}$. During the committing phase, the committer $\CO$ commits to a bit $b \in \{0,1\}$, with $\CO_{com}(1^{\lambda}, b)$ representing this step for the committer and $\RE_{com}(1^{\lambda})$, for the receiver. After the interaction, $\CO_{com}$ outputs a state $\rho^{\CO}_{com} \in \BB_1(\HH_{\CO})$ and $\RE_{com}$, $\rho^{\RE}_{com} \in \BB_1(\HH_{\RE})$. All in all, the commitment phase is written as $(\rho^{\CO}_{com},\rho^{\RE}_{com}) \mapsfrom \langle \CO_{com}(1^{\lambda},b), \RE_{com}(1^{\lambda}) \rangle $. Additionally, $\rho^{\CO}_{com}$ and $\rho^{\RE}_{com}$ may be entangled.  During the decommitment phase, the interaction between $\CO_{dec}(\rho^{\CO}_{com})$ and $\RE_{dec}(\rho^{\RE}_{com})$ leads to an output of $\RE_{dec}$ that is either $b'$ or $\perp$.

\subsubsection{Security against malicious receiver}
The two notions of security that we use against malicious receivers are as follows: 

 \begin{definition}[Computational hiding] A bit commitment protocol $(\CO,\RE)$ is computationally hiding if, given any polynomial-size interactive receiver $\RE^*_{com} = \{\RE^*_{com,\lambda}, \mathbf{\rho}_{\lambda}\}$ with $\text{OUT}_{\RE} \langle \CO(1^{\lambda}, b), \RE^*_{com, \lambda} (\mathbf{\rho}_{\lambda}) \rangle$ being the output of the receiver after the commitment phase, there exists a negligible function $\nu (\cdot)$ such that:

\begin{align*}
 & \Big | \probP [OUT_{\RE}\langle\CO_{com}(1^{\lambda},0),\RE^*_{com}(\rho_{\lambda})\rangle=1] - \\
        &\probP[OUT_{\RE}\langle\CO_{com} (1^{\lambda},1),\RE^*_{com}(\rho_{\lambda})\rangle=1]\Big |=\nu (\lambda).
    \end{align*}
\label{def:comphid}
\end{definition}

\begin{definition} [Equivocality]
    Given a set of auxiliary states $\{\rho_{\lambda}, \sigma_{\lambda}\}_{\lambda \in \NN}$ with $\rho_{\lambda}, \sigma_{\lambda} \in \BB_1(\CC^n)$, a bit commitment protocol $(\CO,\RE)$ is \textit{equivocal} if, for any poly-size receiver $\RE^* = \{\RE_{com, \lambda}^*,\RE_{dec, \lambda}^*, \rho_{\lambda}\}$, there exists a non-uniform quantum polynomial time equivocator $\QQ_{\RE^*} = \{\QQ_{\RE^*,com},\QQ_{\RE^*,dec}\}$ such that, for any poly-size distinguisher $\DD = \{ \DD_{\lambda}, \sigma_{\lambda} \}$, there exists a negligible function $\nu (\lambda) > 0$ for the committed bit $b \in \{0,1\}$:
    
\begin{equation*}
   \abs{ \probP [\DD^* (\sigma_{\lambda}, Real_b) = 1] - \probP [\DD^* (\sigma_{\lambda}, Ideal_b) = 1} = \nu(\lambda),
   %check the function \nu
\end{equation*}

\noindent provided that: 

\smallskip
\noindent \textbf{Real$_b$}: Interaction between $\RE_{com}^*(\rho)$ and $\CO_{com}(1^{\lambda}, b)$ such that $(\rho^{\CO}_{com}, \rho^{\RE^*}_{com}) \mapsfrom \langle \CO_{com}(b),\RE_{com}^*(\rho) \rangle $ and $ \rho^{\RE^*}_{final} \mapsfrom \langle \CO_{dec}(\rho^{\CO}_{com}),\RE_{dec}^*(\rho^{\RE^*}_{com}) \rangle$.

\smallskip
\noindent \textbf{Ideal$_b$}: The quantum simulator outputs $(\rho^{\CO}_{com}, \rho^{\RE^*}_{com}) \mapsfrom \QQ_{\RE^*,com} (\rho)$. \\ Then, $\rho^{\RE^*}_{final}$ $\mapsfrom $ $\langle \QQ_{\RE^*,dec} (\rho^{\CO}_{com}, b),\RE_{dec}^*(\rho^{\RE^*}_{com}) \rangle$.

 \label{def:equivocal}
\end{definition}

\subsubsection{Security against malicious sender}

The security definitions of bit commitment against malicious sender are written as:
\begin{definition} [Statistical binding]
 A bit commitment protocol $(\CO,\RE)$ is statistically-binding if, for every unbounded-size committer $\CO^* = \{\CO_{com}^*,\CO_{dec}^*\}$ and $\lambda \in \NN$, such that:
\begin{align*}
   \left(\rho^{\CO^*}_{com}, \rho^{\RE}_{com}\right) \mapsfrom \left(\langle \CO^*_{com}(1^{{\lambda}}, b),\RE_{com}(1^{{\lambda}}) \rangle\right),
 \end{align*} 
there exists a bit $b \in \{0,1\}$ and a negligible function $\nu(\cdot)$ such that with probability at least $1 - \nu(\lambda)$,
 \begin{equation*}
       \probP[OUT_{\RE} \langle \CO_{dec}^* (1^{ \lambda}, \rho^{\CO^*}_{com} ( \lambda)) ,\RE_{dec} (1^{\lambda}, \rho^{\RE}_{com} ( (\lambda)) \rangle\ \ne  b ] \leq \nu(\lambda).
\end{equation*} 
 \label{def:statistican-binding}
 \end{definition}

In this work, we use a weaker notion of binding that we call relaxed statistical binding.
 Unlike the general definition of statistically-binding bit commitment, each individual commitment is not statistically binding by itself. However, committing to a collection of bits binds for a fraction $1-\eta$ of the bits. 

 \begin{definition} [$\eta$-relaxed statistical binding] 
 Let $(\CO,\RE)$ be a bit commitment protocol.
 For a sequence of $m$ commitments $\left( \CO_i^* = ( \CO_{com,i}^*, \CO_{dec,i}^*)\right)_{i \in m}$ with $\lambda \in \NN$ such that:
\begin{align*}
     \left(\boldsymbol{\rho}^{\CO^*}_{com} ( \lambda),\boldsymbol{\rho}^{\RE}_{com} (\lambda) \right)  \mapsfrom \left(\langle \CO^*_{com, i}(1^{{\lambda}}, b_i),\RE_{com, i}(1^{{\lambda}}) \rangle\right)_{i \in m}.
 \end{align*}

  $(\CO,\RE)$ is $\eta$-relaxed-binding if given the commitment there  exists a set $I \subseteq [m]$, where $|I| \geq (1-\eta)m$ and a negligible function $\nu(\cdot)$ such that with probability at least $1-\nu(\lambda)$ over $\left(\boldsymbol{\rho}^{\CO^*}_{com} ( \lambda),\boldsymbol{\rho}^{\RE}_{com} (\lambda) \right)$, there exists a string $\boldsymbol{b} \in \{0,1\}^{|I|}$ such that,
 \begin{equation*}
       \probP[OUT_{\RE} \langle \CO_{dec}^* (1^{ \lambda}, \boldsymbol{\rho}^{\CO^*}_{com} ( \lambda)) ,\RE_{dec} (1^{\lambda}, \boldsymbol{\rho}^{\RE}_{com} (\lambda)) \rangle|_I \ne \boldsymbol{b} (\lambda) ] \leq \nu(\eta).
\end{equation*} 

 \label{def:relaxbind}

 \end{definition}

Relaxed extractability works similarly to relaxed binding: after committing to many bits, most of them can be extracted.

\begin{definition}[$\chi$-relaxed extractability]
 Let $(\CO,\RE)$ be a bit commitment protocol.
     Given a sequence of auxiliary states  $\left(\rho_i, \sigma_i  \right)_{i \in m}$ where $\rho_i  , \sigma_i  \in \BB_1(\CC^{ {\lambda}})$ and $\lambda \in \NN$, $m$ sequential repetitions of commitments $(\rho^{\CO^*}_{com, 1}, \rho^{\RE}_{com, 1}),$ $...,(\rho^{\CO^*}_{com,  m}, \rho^{\RE}_{com, m})$ are \textit{$\chi$-relaxed-extractable} if, for any poly-size quantum malicious committer $\CO^*$ $ = \{ (\CO^*_{com, i} (1^{{\lambda}}))_{i \in m}, $ $ ( \rho_i )_{i \in m} \}$, there exists a quantum polynomial time extractor $\QQ_{\CO^*} = \{\QQ_{\CO^*,com},\QQ_{\CO^*,dec}\}$ such that, for any poly-size distinguisher $\DD^* = \{ (\DD^*_i)_{i \in m},  ( \sigma_{i})_{i \in m} \}$ and for every polynomial-size opening strategy $\left(\CO^*_{dec, i}\right)_{i \in m }$, there is a negligible function $\nu (\chi) > 0$ that obeys:

\begin{equation*}
   \abs{ \probP [\DD^* (\sigma_{\lambda}, Real) = 1] - \probP [\DD^* (\sigma_{\lambda}, Ideal) = 1]} = \nu(\chi),
\end{equation*}
with: 

\smallskip
\noindent \textbf{Real}: Interaction between $\RE_{dec}(\rho^{\RE}_{com})$ and $\CO^*_{dec}(\rho^{\CO^*}_{com})$ such that $ (\rho^{\CO^*}_{final},b_i)$ $ \mapsfrom \langle \CO^*_{dec}(\rho^{\CO^*}_{com}),\RE_{dec}(\rho^{\RE}_{com}) \rangle_i$, with $b_i \in \{0,1, \perp\}$. Then, the output after $m$ sequential repetitions is given by  
 $\left(\boldsymbol{\rho}^{\CO^*}_{final} (m),\boldsymbol{b} (m) \right) \mapsfrom \left( \langle \CO^*_{com, i}(1^{{\lambda}}, b_i),\RE_{com, i}(1^{{\lambda}}) \rangle\right)_{i \in m}$, where $\boldsymbol{b} \in \{0,1\}^{m}$.

\smallskip

\noindent \textbf{Ideal}: The simulator computes $(\boldsymbol{\rho}^{\CO}_{com} (m), \boldsymbol{\rho}^{\RE}_{com}(m), \boldsymbol{b^*}(m)) \mapsfrom (\QQ_{\CO^*,com} (\rho_i))_{i \in \lambda}$, where $\boldsymbol{b}^* \in \{0,1, \mathbbold{\qm}\}^m$.  Then, $(\boldsymbol{\rho}^{\CO}_{final} (m),\boldsymbol{b} (m))$ $\mapsfrom $ $\langle \CO^*_{dec} (\boldsymbol{\rho}^{\CO}_{com}(m)),\RE_{dec}(\boldsymbol{\rho}^{\RE}_{com} (m)) \rangle$. The simulator outputs FAIL if :
    \begin{itemize}
        \item $|\{i : b^*_i = \mathbbold{\qm}\}| \geq \chi m$, or
        \item if for any $i$ such that $b^*_i  \in \{0,1\}$, $b_i \ne b^*_i$.
    \end{itemize} 
    If the simulator does not output FAIL, it outputs $(\boldsymbol{\rho}^{\CO}_{final}(m),\boldsymbol{b}_S(m))$, where $\boldsymbol{b}_S(j) =\boldsymbol{b^*}(j)$ every time $b^*_i \neq  \mathbbold{\qm}$, and replaces the $\mathbbold{\qm}$ cases by the values $\boldsymbol{b} (j)$ opened by $\CO^*_{dec} (\boldsymbol{\rho}^{\CO}_{com}(j))$. 

\label{def:extract}
\end{definition}

\subsection{Leftover hash lemma} 

\begin{definition}[Quantum conditional min-entropy \cite{Ren05}]
     Given a classical-quantum hybrid state $\rho_{XE} = \sum_{x \in \XX} p_x(x) \ket{x}\bra{x} \otimes \sigma_E^x$, the conditional min-entropy $\text{H}_{\text{min}} ( X |E )$ is defined as:

     \begin{equation*}
         \text{H}_{\text{min}} ( X |E ) = \text{H}_{\text{min}} ( \rho_{XE} |E ) = \sup_{\sigma_E} \max \{h \in \RR : 2^{-h} \, \mathbbm{I} \otimes \sigma_E - \rho_{XE} \geq 0 \}.
     \end{equation*}
\end{definition}

\begin{lemma}[Leftover Hash Lemma with Quantum Side Information \cite{Ren05}]
Let $\rho_{XE}$ be a hybrid state and $h(r,x): \{0,1\}^m \times \{0,1\}^n \rightarrow \{0,1\}^{\ell} $ a 2-universal hash function, with $r$ uniformly distributed over $\RE$. Then, $ K = h(r,x)$ satisfies:

\begin{equation*}
        \Delta (\rho_{K_{\bar{b}}K_{b}E}, \frac{1}{2^{\ell}} \mathbbm{I} \otimes \rho_{K_{b}E}) \leq \frac{1}{2}\sqrt{2^{\ell-H_{min}(X|E)}}.
\end{equation*}
\label{lem:hash}
\end{lemma}

\begin{lemma}[Conditional min-entropy \cite{BF12}]
    Given an $n$-qubit system $A$, let the state $\ket{\psi_{AE}} = \sum_{\substack{\textbf{b} \in \{0,1\}^n \\ \abs{w(\textbf{b})-\alpha} \leq \delta}} \ket{b} \otimes \ket{\psi_E^b} \in \HH_A \otimes \HH_E$ have a relative Hamming weight $\delta$-close to $\alpha$, with $\delta + \alpha \leq 1/2$. Let $X$ be the random variable obtained by measuring $A$ in the bases $H^{\boldsymbol{\theta}}\{\ket{0},\ket{1}\}^{\otimes n}$ for $\boldsymbol{\theta} \in \{0,1\}^n$. Then:

\begin{equation}\label{eq:entropy}
       H_{\text{min}}(X | E) \geq d_H(\boldsymbol{\theta}, \hat{\boldsymbol{\theta}} ) - h_2(\delta + \alpha)n,
\end{equation}
where $\hat{\boldsymbol{\theta}} \in \{0,1\}^n$ denotes the measurement bases committed by Bob, $h_2(\cdot): [0,1] \rightarrow [0,1]$, the binary entropy function and $ d_H(\cdot, \cdot ): \{0,1\}^n \times  \{0,1\}^n \rightarrow \NN$, the Hamming distance function.
    
\end{lemma}

\subsection{Watrous Rewinding Lemma}

\begin{lemma}[Rewinding lemma with small perturbations \cite{Wat05}\label{lem:Wat}]
Let $Q$ be an $(n, k)$-quantum circuit that outputs a classical bit $b$ and $m$ qubits. For any input state $\ket{\psi} \in \BB_1(\CC^n)$, let $p(\psi)$ be the probability of measuring the classical bit $b=0$ such that the state after measuring the action of the  the circuit $Q$ is $ \ket{\phi_0(\psi)} \in \BB_1(\CC^m)$. Let $p_0 \in (0,1)$, $\epsilon \in (0,1/2)$ and $q \in (0,1)$, where $q$ is the probability that the quantum circuit acts like a unitary, be such that:
\begin{enumerate}
    \item $\abs{p(\psi)-q} < \epsilon$, $\forall \ket{\psi} \in \BB_1(\CC^n)$,
    \item $p_0 < p(\psi)$, $\forall \ket{\psi} \in \BB_1(\CC^n)$,
    \item $p_0(1-p_0) \leq q(1-q)$.
\end{enumerate}

Then, there exists a general quantum circuit $\RE$ of size $ \mathcal{O} \left( \frac{\log(1/\epsilon) \abs{Q}}{p_0(1-p_0)} \right)$, such that, for every input state $\ket{\psi}$, the output $R(\ket{\psi})$ satisfies:

\begin{equation*}
    \norm{R(\ket{\psi}) -\ket{\phi_0(\psi)} }_1 \leq 4 \sqrt{\epsilon} \frac{\log(1/\epsilon)}{p_0(1-p_0)} .
\end{equation*}  
\end{lemma}

\subsection{Error resilience}
\label{sec:errorcorr}
Handling error correction in our protocol involves more subtleties compared to, for example, QKD, since in our case Alice and Bob are not trusted. In quantum communication protocols, there are two primary sources of errors: noise and losses. In the case of losses, since the proposed protocol incorporates bit commitments for each BB84 state, the loss of a state will result in the absence of a commitment for that instance of the protocol. Discarding non-committed instances has no impact on the security analysis of the protocol.

The potential presence of malicious parties prevents them from cooperating to estimate the noise. Therefore, an error correction subroutine must be included in the protocol to handle errors without any collaboration between parties \cite{DFSS06,EHG+13,LPAK23}. An $(\mathcal{M}, K, t)$-error correction code can correct errors in a code $C$ composed of $K$ elements from a metric space $\mathcal{M}$, such that for every element $w \in \mathcal{M}$, there exists at most one codeword $c \in C$ within a ball of radius $t$ around $w$.

We propose the use of a non-interactive (classical) error correction code, which is syndrome-based, such as Low-Density Parity-Check (LDPC) codes \cite{Ga60}. The syndrome of an element $w$ is its projection onto the orthogonal subspace of the code, capturing all the necessary information for decoding.

Finally, we consider the possibility of multiphoton events. While an ideal BB84 source generates single photons, experimental sources could accidentally generate more than one copy of a BB84 instance. We take into account this leakage of information for the privacy amplification of our proposed protocol. We present this as a general variable $\vartheta$ that accounts for the information leakage. 

\section{Equivocal Commitment}
\label{sec:eqcomm}

In this section, we present a compiler that yields an equivocal $\eta$-relaxed statistically binding bit commitment from a statistically binding computationally hiding bit commitment. As previously discussed in the Introduction, our compiler is directly based on \cite{BCKM21}, but we reduce the number of repetitions in order to improve efficiency at the cost of requiring relaxed binding properties. 
 The protocol is described in \Cref{alg:equivocalcomp} and we prove the relaxed binding property in \Cref{relaxedstatisticalbinding} and equivocality in \Cref{sec:equivocality-proof}.

\subsection{Relaxed statistical binding}

\label{relaxedstatisticalbinding}

Given a statistically-binding and computationally-hiding commitment, we can build a commitment named \textit{EqCommitment} that is $\eta$-relaxed statistically binding (\Cref{def:relaxbind}) and computationally hiding (\Cref{def:comphid}).
    
\begin{algorithm}[ht!]
    \caption{EqCommitment: Equivocal commitment}
    \textbf{Input of $\CO$:} bit $b \in \{0,1\}$, seeds $\boldsymbol{p}^0_0,\boldsymbol{p}^0_1,\boldsymbol{p}^1_0,\boldsymbol{p}^1_1 \in \{0,1\}^{\lambda_{PQS}}$.\\
    \textbf{Assumptions}: Bit commitment $Com_{r}$ that is statistically binding and computationally hiding and which uses $r$ as the randomness for the commitment.
     
    \begin{algorithmic}[1]
    
    \item[\underline{Commit phase:}]
    \STATE $\CO$ chooses two random bits $u^0, u^1 \in \{0,1\}$.
   
    \STATE $\CO$ and $\RE$ commit to four values, by using the input seeds in the following way:
    \begin{itemize}
        \item $c_0^0$ = $Com_{\boldsymbol{p}^{0}_0} \langle \CO(u^0), \RE \rangle$
        \item $c^0_1$ = $Com_{\boldsymbol{p}^{0}_1} \langle \CO(u^0), \RE \rangle$
        \item $c^1_0$ = $Com_{\boldsymbol{p}^{1}_0} \langle \CO(u^1), \RE \rangle$
        \item $c_1^1$ = $Com_{\boldsymbol{p}^{1}_1} \langle \CO(u^1), \RE \rangle$
    \end{itemize}

    \STATE $\RE$ sends a bit $\gamma \in \{0,1\}$ to $\CO$ 
    \STATE $\RE$ and $\CO$ decommit $u^\gamma_0 = Decom_{\boldsymbol{p}^\gamma_0} (c^{\gamma}_{0})$ 
        and $u^\gamma_1 = Decom_{\boldsymbol{p}^{\gamma}_1} (c^{\gamma}_{1})$ 
    \STATE $\RE$ aborts if $u^\gamma_0 \ne u^\gamma_1$.
    
    \STATE $\CO$ computes $e = b \oplus u^{\bar{\gamma}}$ and sends it to $\RE$.

    \item[\underline{Decommit phase:}]
    \STATE $\CO$ chooses a random bit $\delta \in \{0,1\}$ and sends both, $\delta$ and $b$ to $\RE$.
    \STATE $\CO$ and $\RE$ decommit $c^{\bar{\gamma}}_{\delta}$:
    \begin{equation*}
        u^{\Bar{\gamma}}_{\delta} = Decom_{\boldsymbol{p}_{\delta}^{\bar{\gamma}}}(c^{\bar{\gamma}}_{\delta})
    \end{equation*}
    \STATE $\RE$ aborts if $ u^{\Bar{\gamma}}_{\delta} \ne b \oplus e$.
    \end{algorithmic}    \label{alg:equivocalcomp}
    \end{algorithm}

\begin{theorem}
\label{th:fiststatbind}
If $Com$ is a statistically binding commitment scheme, then, for any function $\eta(n)$, EqCommitment is $\left(\eta,2^{-n\eta}\right)$-relaxed statistically binding.
 \end{theorem}
\begin{proof}
Let us consider $n$ sequential repetitions of the EqCommitment protocol. Let us denote 
$c^{\gamma}_{\delta, i}$ as the commitment 
$c^{\gamma}_{\delta}$ of the $i$-th run of the protocol.  Notice that all $c^{\gamma}_{\delta, i}$ are statistically binding, so with overwhelming probability, we can define the values $u^{\gamma}_{\delta, i}$ such that the opening of $c^{\gamma}_{\delta, i}$  to a value that is different than $u^{\gamma}_{\delta, i}$ fails.

Let us define the set of $B_0 = \{i: \exists {\gamma} \quad u^{{\gamma}}_{0,i} \ne u^{{\gamma}}_{1, i}\}$. Let us also define the event $E_0$ where the check on Step 5 passes on all runs of the protocol. Notice that for each $j \not\in B_0$, the $j$-th run of the commitment is statistically binding, and the goal is now to connect $|B_0|$ and $\probP [E_0]$. 
It can be shown that:
\begin{align}\label{eq:connection-E-B}
    \probP [E_0] \leq \frac{1}{2^{|B_0|}} ,
\end{align}
which is equivalent to the fact that if $\probP[E_0] \geq  2^{-n\eta}$, then $|B_0| \leq n\eta$, which shows that EqCommitment is $(\eta,2^{-n\eta})$-relaxed statistically binding.  The proof of \Cref{eq:connection-E-B} holds since for each $i \in B_0$, there exists one $\gamma_i^*$ such that $u^{\gamma_i^*}_{0,i} \ne u^{\gamma_i^*}_{1,i}$, and in order for $E_0$ to be true, the challenged $\gamma_i$ for that run is different than $\gamma_i^*$, which happens with probability $\frac{1}{2}$. Since all challenges are picked independently, the probability is:
\[\probP [\forall i \in B_0: \gamma_i \ne \gamma^*_i] = \prod_{i \in B_0}\ \probP [\gamma_i \ne \gamma^*_i] \leq \frac{1}{2^{|B_0|}},\]
where the inequality accounts for the trivial case where $\nexists \gamma_i$ such that $ u^{\gamma_i}_{0,i} \ne u^{\gamma_i}_{1,i}$, resulting in $\probP[E_0] = 0$.
\end{proof}

\subsection{Equivocality}

\label{sec:equivocality-proof}

\begin{algorithm}[t]
    \caption{Equivocal simulator $\QQ_{Eq}$}
    \begin{algorithmic}[1]
    \STATE $\QQ_{\RE^*}$ samples a random bit $\Tilde{\gamma} \in \{0,1\}$ and stores it in register P.
     \STATE $\CO$ samples two random bits $\Tilde{u}, \Tilde{y} \in \{0,1\}$
    \STATE Commitment subprotocol for the simulated committer $\Tilde{\CO}$
    \begin{itemize}
    \item If $\Tilde{\gamma} = 0$, $\Tilde{\CO}$ commits to:
    \begin{enumerate}
        \item $\Tilde{c}_0^0$ = $Com_{\boldsymbol{p}_0^0} \langle \Tilde{\CO}(\Tilde{y}), \RE^* \rangle$
        \item $\Tilde{c}^0_1$ = $Com_{\boldsymbol{p}_1^0} \langle \Tilde{\CO}(\Tilde{y}), \RE^* \rangle$
        \item $\Tilde{c}^1_0$ = $Com_{\boldsymbol{p}_0^1} \langle \Tilde{\CO}(\Tilde{u}), \RE^* \rangle$
        \item $\Tilde{c}_1^1$ = $Com_{\boldsymbol{p}_1^1} \langle \Tilde{\CO}(1-\Tilde{u}), \RE^* \rangle$
      \end{enumerate}
      \item If $\Tilde{\gamma} = 1$, $\Tilde{\CO}$ commits to:
    \begin{enumerate}   
        \item $\Tilde{c}_0^0$ = $Com_{\boldsymbol{p}_0^0} \langle \Tilde{\CO}(\Tilde{u}), \RE^* \rangle$
        \item $\Tilde{c}^0_1$ = $Com_{\boldsymbol{p}_1^0} \langle \Tilde{\CO}(1-\Tilde{u}), \RE^* \rangle$
        \item $\Tilde{c}^1_0$ = $Com_{\boldsymbol{p}_0^1} \langle \Tilde{\CO}(\Tilde{y}), \RE^* \rangle$
        \item $\Tilde{c}_1^1$ = $Com_{\boldsymbol{p}_1^1} \langle \Tilde{\CO}(\Tilde{y}), \RE^* \rangle$
      \end{enumerate}
    \item Run commitment phase between $\CO$ and $\RE^*$.
    \item Measure $\gamma$ in the register A. If $\Tilde{\gamma}=\gamma$, the simulator was successful, output 0. Otherwise, output 1. 
    \end{itemize}

    \STATE Let $u = b \oplus e$
    \STATE $\QQ_{\RE^*}$ and $\RE^*$ decommits to:
    \begin{itemize}
        \item If $\Tilde{\gamma} = 0$, decommits to the $(u \oplus \tilde{u} + 2)^{th}$ committed bit.
        \item If $\Tilde{\gamma} = 1$, decommits to the $(u \oplus \tilde{u})^{th}$ committed bit.
    \end{itemize}

 \end{algorithmic}    \label{alg:equivocalsim}
    \end{algorithm}

Our proof of equivocality is the same as \cite{BCKM21}, and we provide it here for completeness. We apply Watrous' rewinding on the equivocal simulator in \Cref{alg:equivocalsim}  to pass the verification test while not having consistent commitments.  

\begin{theorem}
\label{th:firsteq}
If Com is a computationally hiding commitment, then EqCommitment is an equivocal bit commitment
\end{theorem} 
\begin{proof}
    
In order to apply the rewinding lemma, we can consider a communication scheme as follows.
\begin{itemize}
    \item \textbf{$\CO$}: the committer $\CO$ commits and communicates with the malicious receiver $\RE^*$ through the polynomial-size register M.
    \item \textbf{$\RE^*$}: the malicious receiver has three different registers: W, V and A. The first one, W, represents the quantum input. The registers V and A are the work space of $\RE^*$ and both are initialized in the zero state. While V is a polynomial-size register, A is a one qubit register. $\RE^*$ measures this second register A for sampling $\gamma \in \{0,1\}$. After the honesty check, $\RE^*$ outputs the registers (W,V,A,M).
    \item \textbf{$\QQ_{\RE^*}$}: the quantum simulator works with registers P and Z in addition to (W,V,A,M). P is a one qubit register for the guess $\tilde{\gamma}$, and Z is the auxiliary register in which Algorithm \ref{alg:equivocalsim} is implemented.
\end{itemize}

The condition to apply Watrous' lemma is that the probability for the simulator to guess the output has to be nearly independent of the probability distribution of the malicious receiver $\RE^*$'s choice. This is guaranteed by the computationally hiding property of the underlying bit commitment.  Then, $\abs{p(\psi)- 1/2} = negl(\lambda)$, for every input state $\ket{\psi}$ of $\RE^*$ in $\QQ_{Eq}$.  Due to \Cref{lem:Wat}, there is a poly-size circuit $\RE$ that outputs a state inverse-exponentially close to the final state conditioned to $\tilde{\gamma} = \gamma$ of the Algorithm \ref{alg:equivocalsim}.  

Once the simulator $\QQ_{R^*, com}$ has measured the qubit register P and obtained $\gamma$, it proceeds to do the honesty check, as described in Algorithm \ref{alg:equivocalsim} and outputs the registers (W,V,A,M) after tracing out the remaining registers.

If a quantum poly-time distinguisher $\DD^* = \{ \DD^*_\lambda, \sigma_{\lambda} \}$ could distinguish $Real_b$ and $Ideal_b$, it would mean that $\DD^*$ is able to discriminate the commitments of the simulator $\QQ_{R^*}$ and the honest committer $\CO$. This would imply that the distinguisher is able to open the committed values which contradicts the definition of computational hiding of the base commitment scheme. 

\end{proof}

\section{Equivocal and relaxed-extractable commitment}
\label{sec:EREcom}
We construct an equivocal and $\chi$-relaxed extractable commitment (\Cref{def:extract}), named \textit{ERE-Commitment}, based on an equivocal and $\eta$-relaxed statistically binding commitment. We present our protocol in \Cref{alg:EqExcomm} and then prove its equivocality and the relaxed extractability properties.

\begin{algorithm}
    \caption{ERE-Commitment: Equivocal and relaxed extractable commitment}
     \textbf{Committer} $\CO$ \textbf{Input:} A sequence of bits $(b_i)_{i \in [wk]}$. \\
    \textbf{Assumptions}: The parties use a $PRG : \{0,1\}^{\lambda_{PQS}} \rightarrow \{0,1\}^{2w \cdot \lambda_{PQS}}$, 2-universal hash functions $h_j(\cdot, \cdot): \{0,1\}^m \times \{0,1\}^n \rightarrow \{0,1\}^{\lambda_{PQS}}$, an EqCommitment that is equivocal and a statistically binding and computationally hiding commitment $Com_{r}$ which uses $r$ as the randomness of the commitment. 
     
     \label{alg:EqExcomm}
     
    \begin{algorithmic}[1]
    
     \STATE  $\CO$ samples $\boldsymbol{x} \leftarrow \{0,1\}^{4\lambda_{EX}}$ and $\boldsymbol{\theta} \leftarrow \{+,\times\}^{4\lambda_{EX}}$ and sends $\ket{\boldsymbol{x}}_{\boldsymbol{\theta}}$ to $\RE$.
    \STATE $\RE$ chooses his measurement bases $\hat{\boldsymbol{\theta}} \leftarrow \{+,\times\}^{4\lambda_{EX}}$ and measures $\hat{\boldsymbol{x}} \leftarrow \{0,1\}^{4\lambda_{EX}}$. 
\STATE  
$\RE$ sequentially commits to $\left((\hat{x}_i,\hat{\theta}_i)\right)_{i \in [4\lambda_{EX}]}$ using EqCommitment.    

\STATE $\CO$ challenges a random subset $E \subset [4 \lambda_{EX}]$, such that $\abs{E} = 2 \lambda_{EX}$.
\STATE $\CO$ and $\RE$ execute EqDecommitment sequentially for every $i \in E$. Then, $\CO$ obtains  $\left((\hat{x}_i, \hat{\theta}_i)\right)_{i \in E}$. $\CO$ aborts if the commitment fails to open.
\STATE $\CO$ checks that $x_i = \hat{x}_i$ whenever $\theta_i = \hat{\theta}_i$, up to a fraction $\alpha$ due to the experimental error. If the test does not pass, $\CO$ aborts.

\STATE $\CO$ divides $\bar{E}$ into $2k$ consecutive disjoint subsets $M_j \subseteq \Bar{E}$ of size $m$, with $m$ given by the size of the preimage of a 2-universal hash function $\{h_j(\cdot, \cdot): \{0,1\}^m \times \{0,1\}^n \rightarrow \{0,1\}^{\lambda_{PQS}}\}_{j\in [2k]}$, where $k = \lfloor\lambda_{EX}/m\rfloor$.

\STATE For every $j \in [2k]$, $\CO$ samples $\boldsymbol{r}^j \in \{0,1\}^{n}$ and compute $\boldsymbol{s}^j = h_j(\boldsymbol{r}^j, \Tilde{\mathbf{x}}^j)$ where $\Tilde{\mathbf{x}}^j = \{x_i : i \in M_j \}$.  
        
\STATE Given $PRG : \{0,1\}^{\lambda_{PQS}} \rightarrow \{0,1\}^{2w \cdot \lambda_{PQS}}$, $\CO$ generates $\boldsymbol{p}^j = PRG(\boldsymbol{s}^j) = (\boldsymbol{p}^{j}_1,...,\boldsymbol{p}^j_{2w})$ for every $j \in [2k]$, where  $\boldsymbol{p}^j_i \in \{0,1\}^{\lambda_{PQS}}$.

\STATE $\CO$ sends $(M_j)_{j \in [2k]}$, $\boldsymbol{\theta}$ and  $(\boldsymbol{r}^j)_{j \in [2k]}$  to $\RE$. $\CO$ sends the syndromes $(Synd_j)_{j \in [2k]}$ to correct a proportion $\alpha$ of errors.

\STATE  $\CO$ and $\RE$ sequentially execute $k$ sessions of $w$ parallel EqCommitment. For each $r \in [k]$:
\begin{enumerate}[label*=11.\arabic*:]
    \item $\CO$ commits in parallel to $w$ pairs of random bits with $\left(\boldsymbol{p}_{i_{q,0}}^{j^{r,0}},\boldsymbol{p}_{i_{q,1}}^{j^{r,0}}, \boldsymbol{p}_{i_{q,0}}^{j^{r,1}},\boldsymbol{p}_{i_{q,1}}^{j^{r,1}} \right)_{q \in [w]} = \left(\boldsymbol{p}_{2q-1}^{2r-1},\boldsymbol{p}_{2q}^{2r-1}, \boldsymbol{p}_{2q-1}^{2r},\boldsymbol{p}_{2q}^{2r} \right)_{q \in [w]}$ as seeds.
    \item For the honesty check of the EqCommitment, $\RE$ samples $\gamma_r \xleftarrow{\$} \{0,1\}$ and sends it to $\CO$.
    \item $\CO$ sends  $\tilde{\mathbf{x}}^{j^{r,\gamma_r}}$ and $p^{j^{r,\gamma_r}}$ to $\RE$. $\RE$ checks if,  
    \begin{itemize}
        \item $\tilde{x}_i = \hat{x}_i$ whenever $\tilde{\theta}_i = \hat{\theta}_i$, and, if there is a proportion of  $\tilde{x}_i \neq \hat{x}_i$ larger than $\alpha$, it aborts. 
        \item $PRG(h_{j^{r,\gamma_r}}(\Tilde{\mathbf{x}}^{j^{r,\gamma_r}})) = \boldsymbol{p}^{j^{r,\gamma_r}}$  and, if not, aborts.
        \item $Decom_{\boldsymbol{p}_{i_{q,0}}^{j^{r,\gamma_r}}} (c_{i_{q,0}}^{j^{r,\gamma_r}})=Decom_{\boldsymbol{p}_{i_{q,1}}^{j^{r,\gamma_r}}}(c_{i_{q,1}}^{j^{r,\gamma_r}})$ for every $q \in [w]$ and, if not, aborts.
    \end{itemize}
    \item   $\CO$ commits to each $(b_i)_{i \in [w]}$ as described in EqCommitment.
\end{enumerate}

    \STATE When necessary, $\CO$ and $\RE$ open to any of the committed bits $(b_i)_{i \in [wk]}$ by executing EqDecommitment. In order to do so, $\CO$ sends the seeds $\boldsymbol{p}_{i_{q,\delta}}^{j^{r,\bar{\gamma}}}$ to $\RE$.

    \end{algorithmic}   
\end{algorithm} 

\subsection{Equivocality}

The proof of equivocality is very similar to \Cref{th:firsteq}. Nevertheless, we also provide in \Cref{col:EqEreCom} the condition for the Naor's commitment subprotocol of the ERE-Commitment to be computationally hiding, giving in \Cref{th:secmalrec} an analytical expression that quantifies it.

\begin{theorem}
If Com is computationally hiding, then ERE-Commitment is an equivocal bit commitment.
\end{theorem} \label{th:secondeq}
\begin{proof}
The main difference with respect to the proof for EqCommitment is that the equivocal simulator has to answer correctly the commitment of $w$ equivocal boxes in parallel with single challenge bit $\gamma_r$.  Given a computationally hiding bit commitment, $\abs{p(\psi)- 1/2} = negl(\lambda)$ also for this case. Hence, all the properties of the Algorithm \ref{alg:equivocalsim} and the application of Watrous' rewinding lemma are equivalent for this scheme.  Therefore, from computational hiding, $Real_b$ and $Ideal_b$ of Def. \ref{def:equivocal} are indistinguishable.

\end{proof}

\begin{corollary}
    If $Com_{\boldsymbol{p}^j_i}$ is implemented using Naor's bit commitment  with input seeds $\boldsymbol{p}^j_i \in \{0,1\}^{\lambda_{PQS}}$, where $(\boldsymbol{p}^j_1,...,\boldsymbol{p}^j_{2w}) = PRG(\boldsymbol{s}^j)$  for $\boldsymbol{s}^j \in \{0,1\}^{\lambda_{PQS}}$ , then ERE-Commitment is an equivocal bit commitment.
\label{col:EqEreCom}
\end{corollary}

\begin{proof}
     We have to prove that the  Naor's bit commitment scheme \cite{Naor91} is computationally hiding giving the pseudorandom seeds $\boldsymbol{p}^j_i \in \{0,1\}^{\lambda_{PQS}}$. Naor's bit commitment is an interactive bit commitment in which the receiver sends $ k \xleftarrow{\$} \{0,1\}^{3\lambda}$ to the committer, which commits to the bit $b$ by sending $(b \cdot k) \oplus W(r) $, where $r \xleftarrow{\$} \{0,1\}^{\lambda}$ and $W : \{0,1\}^{\lambda_{PQS}} \rightarrow \{0,1\}^{3\lambda}$ is a pseudo-random generator. Thus, the computational hiding property is guaranteed by the pseudorandomness property of $W$. In this case, the seed $\boldsymbol{p}^j_i$ of the Naor's base commitment is not randomly sampled by the committer, but obtained as the output of a PRG which input is $(\boldsymbol{s}^j)_{j \in [2k]}$.  $(\boldsymbol{s}^j)_{j \in [2k]}$ are generated by hashing the strings $(\tilde{\boldsymbol{x}}^j)_{ j \in [2k]}$. Therefore, computational hiding would be fulfilled if the strings $(\boldsymbol{s}^j)_{j \in [2k]}$ are random strings. By applying the leftover hash lemma (\Cref{lem:hash}), the obtained seeds are arbitrarily close to a random string, as proven in  \Cref{th:secmalsen} of  \Cref{appendix}.    
\end{proof}

\subsection{Relaxed extractability}

Before proving relaxed-extractability, we first show that our commitment scheme is relaxed-binding.

\begin{theorem}
 If $Com$ is a statistically binding commitment scheme, then  for any function $\eta(k)$ ERE-Commitment is $\eta$-relaxed statistically binding.
 \label{th:secondstatbind}
 \end{theorem}

\begin{proof}
The proof is similar to the one of  \Cref{th:fiststatbind}. In ERE-Commitment, the seeds of the base commitments $Com_r$ are used as described in Algorithm \ref{alg:EqExcomm}. Therefore, when the honesty check is done, the sampled population is not the number of sequential repetitions of the EqCommitment $wk$ but the number of different families $k$ since the EqCommitment instances which make use of seeds derived from the same family $\boldsymbol{p}^j$ are opened together. 

 \end{proof}

To prove extractability, we will use the equivocality properties of EqCommitment.  To prove that the simulator $\QQ_{\CO^*}(\rho)$ defined in \Cref{alg:extractor} that represents the $Ideal$ interaction from  \Cref{def:extract} is indistinguishable from the $Real$ one, we use the following sequence of hybrid arguments:
\begin{itemize}
    \item \textbf{Hybrid 0}. Corresponds to the real protocol described in Algorithm \ref{alg:EqExcomm}.
    
    \item \textbf{Hybrid 1}. Similar to $\text{Hyb}_0$ except in Step 3, in which we execute the equivocal simulator $\QQ_{\RE^*,com}$ for $i \in  [4\lambda_{EX}]$, and in Step 5, to open the challenged commitments $i \in E$, we execute $\QQ_{\RE^*,dec}$.
    
    \item \textbf{Hybrid 2}. Similar to  $\text{Hyb}_1$, except that  $\ket{\boldsymbol{x}}_{\boldsymbol{\theta}}$ is not measured during Step 2 of \ref{alg:EqExcomm}, but rather during Step 3 of Algorithm \ref{alg:extractor}, i.e., only the positions $i \in E$ are measured and only after obtaining the sequence of indices.

    \item \textbf{Hybrid 3}. Identical to $\text{Hyb}_2$ but the simulator performs the error correction as described in Step 7 of \ref{alg:extractor} and then, the honesty check as described in Step 8 of Algorithm \ref{alg:extractor}.

    \item \textbf{Hybrid 4}. Identical to $\text{Hyb}_3$ but the simulator extracts the non-challenged commitments as described in Step 9 of \ref{alg:extractor}. 
    
    \item \textbf{Hybrid 5}. The $Ideal$ distribution, i.e., the extractable simulator described in Algorithm \ref{alg:extractor}. The difference with respect to $\text{Hyb}_4$ is that the simulator opens the non-extracted commitments to the value opened by the malicious committer, as described on Step 10. 
\end{itemize}

\begin{lemma} There exists a negligible function $\nu(\lambda)$, such that:
\begin{equation*}
    \abs{\probP[\DD^*(\sigma,\text{Hyb}_1) = 1] - \probP[\DD^*(\sigma,\text{Hyb}_0) = 1]} = \nu(\lambda).
\end{equation*}
\label{cl:lemma1}
\end{lemma}
\begin{proof}
    Suppose that there exists a non-negligible function for which $\DD^*$ can distinguish between $\text{Hyb}_0$ and $\text{Hyb}_1$. This is equivalent to:

\begin{equation*}
    \abs{\probP[\DD^*(\sigma,\text{Hyb}_{0, i-1}) = 1] - \probP[\DD^*(\sigma,\text{Hyb}_{0,i}) = 1]} \geq \frac{1}{poly(\lambda) 4\lambda },
\end{equation*}
 where $\text{Hyb}_{0, j}$ is the sub-hybrid that is identical to $\text{Hyb}_{0}$ until the committed bit $j \in [4 \lambda_{EX}]$ and then identical to $\text{Hyb}_{1}$. Therefore, $\DD^*$ would be able to distinguish between the outputs $\rho^{CO}_{final}$ of an honest and an equivocated commitment, contradicting \Cref{def:equivocal}. 

\end{proof}

\begin{lemma} 
 $   
    \abs{\probP[\DD^*(\sigma,\text{Hyb}_2) = 1] \equiv \probP[\DD^*(\sigma,\text{Hyb}_1) = 1]}$.
\label{cl:eq2}
\end{lemma}

\begin{proof}
    
    Since the states corresponding to the subset $E \subset [4\lambda_{EX}]$ are also measured with  randomly sampled bases $(\hat{\theta}_i)_{i \in E}$, the two experiments have exactly the same output distribution. 
\end{proof}

\begin{lemma} There exists a negligible function $\nu(\lambda)$, such that,
    
\begin{equation*}
    \abs{\probP[\DD^*(\sigma,\text{Hyb}_3) = 1] - \probP[\DD^*(\sigma,\text{Hyb}_2) = 1]} = \nu(\lambda).
\end{equation*}

\end{lemma}
\begin{proof}

A quantum poly-time distinguisher $\DD^*$ would be able to discriminate between both distributions if $Hyb_2$ and $Hyb_3$ had different probabilities for output $\perp $ during the honesty check. If $\delta > \alpha$, both Hybrids abort. Otherwise, as discussed in \Cref{sec:errorcorr}, an $(\mathcal{M}, 4\lambda_{EX}, 4\alpha \lambda_{EX})$-error correction code can unambiguously correct up to a proportion of $\alpha$ errors. Therefore, for any $\delta \leq \alpha$, there is at most one codeword inside the radius of the error correction code. Therefore, both distributions are indistinguishable.

\end{proof}

\begin{algorithm}
    \caption{Extractable simulator}
    \begin{algorithmic}[1]
    
     \STATE  $\CO^*(\rho)$ outputs $\mathbf{\Psi}.$

    \STATE For $i \in [4\lambda_{EX}]$, execute sequentially the quantum equivocal simulator $\QQ_{\RE^*}$ to generate a tuple of  $((c^{\bar{\gamma}}_{\delta})^{i}_{\delta =0,1})_{i \in [4\lambda_{EX}]}$ dummy commitments. After each commitment session, $(\rho^{\Ss}_i,\rho^y_{com,i}) \mapsfrom \QQ_{\RE^*,com} (\rho)$, where $\rho^{\Ss}_i$ is stored by $\QQ_{\CO^*} (\rho)$, and $\boldsymbol{\rho}^y_{com, i}$ is used by the malicious committer $\CO^*$ as an input in the next session.
    \STATE Once $\CO^*$ has announced the challenged subset $E \subset [4\lambda_{EX}]$, $\QQ_{\CO^*} (\rho)$ samples $\hat{\boldsymbol{\theta}} \leftarrow \{+,\times\}^{2\lambda_{EX}}$. $\QQ_{\CO^*}$ measures the $i^{th}$ qubit of $\mathbf{\Psi}$, $\Psi_i$, for every $i \in E$ in the previous sampled bases, obtaining $(\hat{x}_i)_{i \in E}$.
    \STATE  By applying the equivocal simulator $\QQ_{Eq}$, $\QQ_{\CO^*}$ opens the challenged bits $z_i \in E$ to the correct values $(\hat{x}_i, \hat{\theta}_i)_{i \in E}$. To do so, it executes  for each $i \in E$:  $\rho^x_{final}$ $\mapsfrom $ $\langle \QQ_{\RE^*,dec} (\hat{x}_i, \hat{\theta}_i,z_i,\rho^x_{com}),\CO_{dec}^*(\rho^y_{com}) \rangle$, where $\rho^x_{com}$ is the current state of $\CO^*$, updated after each session.
    \STATE If $\CO^*$ aborts at any point, $\QQ_{\CO^*} (\rho)$ outputs $\perp$, otherwise continues.

     \STATE $\CO^*$ and $\RE$ discard the tested positions and both reorder their respective tuples $((x_i, \theta_i))_{i \in \Bar{E}}$ and $((\hat{x}_i, \hat{\theta}_i))_{i \in \Bar{E}}$. $\CO^*$ sends $\boldsymbol{\theta}$ , computes the $(\boldsymbol{p}^j)_{j \in [2k]}$ as described in Algorithm \ref{alg:EqExcomm} and sends $(\boldsymbol{r}^j)_{j \in [2k]}$.

     \STATE $\QQ_{\CO^*} (\rho)$ measures each  qubit $\mathbf{\psi}_i \in \bar{E}$ in the bases $\theta_i$ to obtain $\hat{x}_i = x_i$.  $\QQ_{\CO^*} (\rho)$ performs the error correction on $\hat{\boldsymbol{x}}$ by using the syndrome $Synd$. Then, $\QQ_{\CO^*} (\rho)$ computes the seeds $\left(\left(\hat{\boldsymbol{p}}_{i_{q,0}}^{j^{r,0}},\hat{\boldsymbol{p}}_{i_{q,1}}^{j^{r,0}}, \hat{\boldsymbol{p}}_{i_{q,0}}^{j^{r,1}},\hat{\boldsymbol{p}}_{i_{q,1}}^{j^{r,1}} \right)_{q \in [w]} \right)_{r \in [k]} $.

     \STATE $\QQ_{\CO^*} (\rho)$ executes the honesty check as described in Step 11 of \Cref{alg:EqExcomm}. 
     
   \STATE $\QQ_{\CO^*} (\rho)$ extracts the non challenged commitments:
   \begin{enumerate}[label*=9.\arabic*:]
   \item $\QQ_{\CO^*} (\rho)$ checks if $Decom_{\hat{\boldsymbol{p}}_{i_{q,0}}^{j^{r,\gamma}}} (c_{i_{q,0}}^{j^{r,\gamma}})=Decom_{\hat{\boldsymbol{p}}_{i_{q,1}}^{j^{r,\gamma}}} (c_{i_{q,1}}^{j^{r,\gamma}})$ where $\hat{\boldsymbol{p}}_{i_{q,\delta}}^{j^{r,\gamma}} $ are the ones obtained in Step 7.
   \item For the positions where this is the case, $b_{r,q}^*= Decom_{\hat{\boldsymbol{p}}_{i_{q,0}}^{j^{r,\gamma}}} (c_{i_{q,0}}^{j^{r,\gamma}} \oplus e_{r,q})$, otherwise $b^*_{r,q} = \mathbbold{\qm}$. 
   \item If the proportion of $b^*_i = \mathbbold{\qm}$ is larger than $\chi = \eta +\zeta $, with $\eta = \zeta = \log^2(k)$, $\QQ_{\CO^*} (\rho)$  outputs aborts, where $i = (r-1)*w+q$ are the reordered indexes. 
   \item Otherwise, $\QQ_{\CO^*} (\rho)$ outputs $(\boldsymbol{\rho}^{\CO}_{com},\boldsymbol{\rho}^{\RE}_{com}, \boldsymbol{b}^*)$, where $\boldsymbol{b}^* \in \{0,1\}^{wk}$  is the string of bits in which $\boldsymbol{b}_i^*= b_i^*$ for the positions in which the output was not $\mathbbold{\qm}$ and an equivocated dummy commitment for the remaining positions by executing the quantum equivocal simulator $\QQ_{Eq}$. $\boldsymbol{\rho}^{\CO}_{com}$ is the resulting state of $\CO^*$ and $\boldsymbol{\rho}^{\RE}_{com} = ( \boldsymbol{\theta}, \hat{\boldsymbol{\theta}}, \hat{\boldsymbol{x}})$.
\end{enumerate}
\STATE $\QQ_{\CO^*} (\rho)$ opens the equivocated commitments corresponding to $\mathbbold{\qm}$ to the value opened by $\CO^*$ in these positions by applying the equivocal simulator $\QQ_{Eq}$.
    
  \end{algorithmic}   \label{alg:extractor}
    \end{algorithm} 

\begin{lemma} There exists a negligible function $\nu(\lambda)$ such that:
    
\begin{equation*}
    \abs{\probP[\DD^*(\sigma,\text{Hyb}_4) = 1] - \probP[\DD^*(\sigma,\text{Hyb}_3) = 1]} = \nu(\lambda). 
\end{equation*}

\end{lemma}

\begin{proof}
The only difference between $\text{Hyb}_3$ and $\text{Hyb}_4$ comes from the fact that the $Ideal$ distribution, given by $\text{Hyb}_4$, outputs FAIL whenever $b_i^* \notin \{ \perp, b_i, \mathbbold{\qm}\}$ and the number of $\mathbbold{\qm}$ outputs is larger than the proportion $\chi \leq \eta + \zeta$, where $\zeta$ is the proportion of commitments seeded in a dishonest way. 

In order to bound the probability of having more than a proportion $\zeta$ of dishonestly seeded commitments, let us consider again $k$ sequential repetitions of $w$ parallel EqCommitment protocols, as described in  \Cref{alg:EqExcomm}. For the $r$-th instance, the $q$-th EqCommiment $\left(c_{i_{q,0}}^{j^{r,0}},c_{i_{q,1}}^{j^{r,0}}, c_{i_{q,0}}^{j^{r,1}}, c_{i_{q,1}}^{j^{r,1}}\right)$ of the $w$ parallel commitments, uses as seeds $ \left(\boldsymbol{p}_{i_{q,0}}^{j^{r,0}},\boldsymbol{p}_{i_{q,1}}^{j^{r,0}},\boldsymbol{p}_{i_{q,0}}^{j^{r,1}},\boldsymbol{p}_{i_{q,1}}^{j^{r,1}} \right)$. 

Each $r \in [k]$ sequential honesty check is composed of a sequence of $w$ EqCommitment that use the seeds $\boldsymbol{p}^{j^{r,0}}\leftarrow\tilde{\boldsymbol{x}}^{j^{r,0}}$ and $\boldsymbol{p}^{j^{r,1}}\leftarrow\tilde{\boldsymbol{x}}^{j^{r,1}}$. Let us define $B_1$ as the set of $r$ indexes with at least one dishonestly seeded commitment $B_1 = \{r: \exists { \gamma,\delta, q} \quad$ $\boldsymbol{p}_{i_{q,\delta}}^{j^{r,\gamma}}  \text{is not the seed of the commitment}\}$. Let us also define the event $E_1$ where the check on Step 11 of Algorithm \ref{alg:EqExcomm} passes on all $k$ sets. The goal is now to connect $|B_1|$ and $\probP [E_1]$.
It can be shown that:
\begin{align}\label{eq:connection-E1-B1}
    \probP [E_1] \leq \frac{1}{2^{|B_1|}} ,
\end{align}
which is equivalent to the fact that, if $\probP[E_1] \geq 2^{-\zeta k}$, then $|B_1| \leq \zeta k$. The proof of \Cref{eq:connection-E1-B1} holds since, for each $r \in B_1$, there exists one set of indexes $\{\gamma^*, \delta^*, q^*\}$ such that $\boldsymbol{p}_{i_{q^*,\delta^*}}^{j^{r,\gamma^*}}$ is not the seed of the commitment, and in order for $E_1$ to be true, the challenged family of commitments seeded by $\boldsymbol{p}^{j^{r,\gamma}}$ for the $r$-th set of $w$ EqCommitment instances is different from $\boldsymbol{p}^{j^{r,\gamma^*}}$, which happens with probability $\frac{1}{2}$. Since all challenges are picked independently, the probability that \[\probP [\forall r \in B_1: \gamma  \ne \gamma^*] = \prod_{r \in B_1}\ \probP [\gamma \ne \gamma^*] \leq \frac{1}{2^{|B_1|}}.\]

 $Hyb_3$ and $Hyb_4$ would be distinguishable if the malicious committer $\CO^*$ has committed to $b_i$ with an equivocated commitment more than a proportion $\eta = \omega(\log^2(k)/k)$, or if the malicious committer $\CO^*$ uses as seeds $\boldsymbol{p}^{j} \neq PRG(h(\Tilde{\mathbf{x}}^j))$ in more than a proportion $\zeta = \omega(\log^2(k)/k)$ of the keys and $\text{Hyb}_3$ has not output $\perp$ in the honesty check phase. This happens with probability: 

\begin{equation*}
\begin{split}
    \probP[\text{FAIL} | \text{Ideal}] =& \probP [ \text{proportion of non extractable  commitments} \geq \eta + \zeta]\\
    &\leq \probP [ \text{ proportion of non extractable commitments} \geq \zeta] .  
\end{split}
\end{equation*}
Then, $\probP[\text{FAIL} | \text{Ideal}] < negl (\lambda)$ given \Cref{th:secondstatbind}. 
\end{proof}

\begin{lemma}
$    \abs{\probP[\DD^*(\sigma,\text{Hyb}_5) = 1] \equiv \probP[\DD^*(\sigma,\text{Hyb}_4) = 1]}$.
\end{lemma}

\begin{proof}
    Similar to the proof of \Cref{cl:lemma1}. In the opening phase, since the opened value is the same as the one opened by the malicious committer $\CO^*$, both distributions are indistinguishable.     
\end{proof}

 Thus, ERE-Commitment is equivocal and $\omega(\log^2(k)/k)$-relaxed-extractable in the sense of \Cref{def:equivocal,def:extract}.

\section{OT from equivocal and relaxed-extractable bit commitment}
\label{sec:QOT}

In this section, we build an oblivious transfer functionality assuming an equivocal and $\chi$-relaxed extractable bit commitment, with $\chi = \omega(\log^2(k)/k)$. The structure is similar to \cite{DFL+09,BCKM21}, but applying it to our setting. 
Lastly, we show how to optimize the number of runs of OT in \Cref{col:distillation}. 

\begin{algorithm}[t]
    \caption{Oblivious transfer protocol}
    \textbf{Alice} $A$ \textbf{Input:} Messages $m_0, m_1 \in \{0,1\}^{\lambda}$.\\ 
    \textbf{Bob} $B$ \textbf{Input:} Bit  $b \in \{0,1\}$.\\
     \textbf{Assumptions}: The parties use an ERE-Commitment that is $\chi$-relaxed extractable and equivocal, $PRG : \{0,1\}^{\lambda_{PQS}} \rightarrow \{0,1\}^{w}$ and 2-universal hash functions $h_j(\cdot, \cdot): \{0,1\}^m \times \{0,1\}^n \rightarrow \{0,1\}^{\lambda_{PQS}}$. 
     
    \begin{algorithmic}[1]
    \STATE  $A$ chooses $\boldsymbol{x} \leftarrow \{0,1\}^{2\lambda_{OT}}$ and $\boldsymbol{\theta} \leftarrow \{+,\times\}^{2\lambda_{OT}}$ and sends $\ket{x}_{\theta}$ to $B$.
    \STATE $B$ samples his measurement bases $\hat{\boldsymbol{\theta}} \leftarrow \{+,\times\}^{2\lambda_{OT}}$ and measures $\hat{\boldsymbol{x}} \leftarrow \{0,1\}^{2\lambda_{OT}}$. $A$ and $B$ execute 
 ERE-Commitment to commit
    to $\left((\hat{x}_i, \hat{\theta}_i)\right)_{i \in [2\lambda_{OT}]}$.
    \STATE $A$ challenges a random subset $T \subset [2 \lambda_{OT}]$, such that $\abs{T} = \lambda_{OT}$.
    \STATE $A$ and $B$ execute ERE-Decommitment for every $i \in T$. Then, $A$ obtains $\left((\hat{x}_i, \hat{\theta}_i)\right)_{i \in T}$.
    \STATE $A$ checks that $x_i = \hat{x}_i$ whenever $\theta_i = \hat{\theta}_i$, up to an experimental error $\alpha $. If the test does not pass, $A$ aborts.
    \STATE $A$ and $B$ discard the tested positions and both reorder their respective sequences $\left((x_i, \theta_i)\right)_{i \in \bar{T}}$ and $\left((\hat{x}_i, \hat{\theta}_i)\right)_{i \in \bar{T}}$. $A$ sends $\boldsymbol{\theta}$ to $B$.
    \STATE $B$ partitions the set $\bar{T}$ into two different subsets: the subset in which both have measured in the same bases $I_b = \{ i : \theta_i = \hat{\theta}_i\}$ and the subset in which they have not measured in the same bases $I_{\bar{b}} = \{ i : \theta_i \neq \hat{\theta}_i\}$. $B$ sends $(I_0, I_1)$ to $A$.
    \STATE $A$ sends the syndromes $Synd_0$ and $Synd_1$ for correcting a proportion $\alpha$ of errors of the bit strings $\boldsymbol{x}_0$ and $\boldsymbol{x}_1$ to $B$, where $\boldsymbol{x}_0 = \{x_i : i \in I_0\}$ and $\boldsymbol{x}_1 = \{x_i : i \in I_1\}$.
    \STATE $A$ samples seeds $s_0,s_1 \leftarrow \{0,1\}^{k(\lambda_{PQS})}$ and sends $(s_0, PRG(h(s_0,\boldsymbol{x}_0)) \oplus m_0, s_1, PRG(h(s_1,\boldsymbol{x}_1)) \oplus m_1)$. 
    \STATE $B$ performs error correction on $\boldsymbol{x}_b$ using the syndrome $Synd_b$ and decrypts $m_b$ by using $\hat{\boldsymbol{x}}_b$, with $\hat{\boldsymbol{x}}_b = \{\hat{x}_i : i \in I_b\}$. If the error correcting protocol fails, $B$ chooses $\boldsymbol{x}_b$ to be a random string $s \leftarrow \{0,1\}^{\lambda_{OT}/2}$.
    \end{algorithmic}    \label{alg:OTProtocol}
    \end{algorithm} 
The oblivious transfer functionality $\FUNC_{OT}(\cdot, \cdot)$ is a quantum two-parties interactive functionality in which Alice, $A$, inputs two messages $m_0, m_1 \in \{0,1\}^{\lambda}$, and Bob, $B$, inputs a choice bit $b \in \{0,1\}$. The output after the interaction is $(m_0,m_1)$ on Alice's side and $(b,m_b)$ on Bob's side.  More concretely, the OT functionality $\FUNC_{OT}$ works as follows:
\begin{itemize}
    \item $\FUNC_{OT}((m_0,m_1), \cdot)$ takes an input $b$ (or $\perp$), returns $m_b$ to $B$ and END (or $\perp$) to $A$.

    \item $\FUNC_{OT}(\cdot, b)$ takes an input $(m_0,m_1)$ (or $\perp$), returns END to $A$ and $m_b$ (or $\perp$) to $B$.
\end{itemize}

 \begin{definition}[Secure OT functionality] A protocol $\langle A(m_0,m_1), B(b) \rangle $ that implements the ideal OT functionality $\FUNC_{OT}(\cdot, \cdot)$ is said to be simulation-based secure if it satisfies:
\begin{itemize}
    \item \textbf{Security against a malicious Alice.} Given the interaction between a quantum poly-time non-uniform malicious Alice, $A^*$, and an honest Bob, their final state will be $\mathbf{\rho}_{OUT,A^*}\langle A^*(m_0,m_1), B(b) \rangle $ and  $\text{OUT}_B\langle A^*(m_0,m_1),$ $ B(b) \rangle $.  There exists a simulator $\text{Sim}_{A^*}$ that interacts with the ideal functionality $\FUNC_{OT}(\cdot, b)$ in the following way. For every non-uniform advice $\rho, \sigma \in \BB_1(\HH)$, where $\rho $ and $\sigma$ may be entangled, $\text{Sim}_{A^*} (\rho)$ sends either $(m_0, m_1)$ or $\perp$ to $\FUNC_{OT}(\cdot, b)$ and outputs the final state $\mathbf{\rho}_{SIM,OUT,A^*}$. The output of the ideal functionality to the Bob is given by $\text{OUT}_B$. Then, for any quantum poly-time non-uniform distinguisher $\DD^*$:

\begin{equation}
    \begin{split}
        &\left| \probP [\DD^*(\mathbf{\sigma}, (\mathbf{\rho}_{SIM,\text{OUT},A^*},\text{OUT}_B)) = 1] \right. \\
        &\left. - \probP[\DD^*(\mathbf{\sigma}, (\mathbf{\rho}_{\text{OUT},A^*}\langle A^*(m_0,m_1), B(b) \rangle,\text{OUT}_B\langle A^*(m_0,m_1), B(b) \rangle)) = 1] \right| \\
        &= \text{negl}(\lambda)
    \end{split} \label{eq:malsend}
\end{equation}

    \item \textbf{Security against a malicious Bob.} Given the interaction between an honest Alice and a quantum poly-time non-uniform malicious Bob, $B^*$, their final state will be  $\text{OUT}_A\langle A(m_0,m_1), B^*(b) \rangle $ and $\mathbf{\rho}_{OUT,B^*}\langle A(m_0,m_1), B^*(b) \rangle $.  There exists a simulator $\text{Sim}_{B^*}$ that interacts with the ideal functionality $\FUNC_{OT}((m_0,m_1), \cdot)$ in the following way. For every non-uniform advice $\rho, \sigma \in \BB_1(\HH)$, where $\rho $ and $\sigma$ may be entangled, $\text{Sim}_{B^*} (\rho)$ sends either $b \in \{0,1\}$ or $\perp$ to $\FUNC_{OT}((m_0,m_1), \cdot)$ and outputs the final state $\mathbf{\rho}_{SIM,OUT,B^*}$. The output of the ideal functionality to Bob is given by $\text{OUT}_A$. Then, for any quantum poly-time non-uniform distinguisher $\DD^*$:

\begin{equation}
    \begin{split}
        &\left| \probP [\DD^*(\mathbf{\sigma}, (\text{OUT}_A,\mathbf{\rho}_{SIM,\text{OUT},B^*} ) = 1] \right. \\
        &\left. - \probP[\DD^*(\mathbf{\sigma}, (\text{OUT}_A\langle A(m_0,m_1), B^*(b) \rangle),\mathbf{\rho}_{\text{OUT},B^*}\langle A(m_0,m_1), B^*(b) \rangle ) = 1] \right| \\
        &= \text{negl}(\lambda).
    \end{split} \label{eq:malrec}
\end{equation}

\end{itemize}
     \label{def:secureOT}
 \end{definition}

\subsection{Security against a malicious Alice}

We will now establish the security of Bob against a quantum poly-time non-uniform malicious Alice, given \Cref{def:secureOT}.  To prove that Eq. \ref{eq:malsend} holds in the proposed scheme, the following hybrid argument is required:

\begin{itemize}
    \item \textbf{Hybrid 0}. The $Real$ distribution given by Algorithm \ref{alg:OTProtocol}.
    \item \textbf{Hybrid 1}. Similar to $\text{Hyb}_0$ except in Step 2, where we execute the equivocal simulator $\QQ_{\RE^*,com}$ for $i \in  [2\lambda_{OT}]$. The second difference with respect to  $\text{Hyb}_0$ is that in Step 4, to open the challenged commitments $i \in T$, we execute $\QQ_{\RE^*,dec}$.
    \item \textbf{Hybrid 2}. The $Ideal$ distribution given by Algorithm \ref{alg:OTsimA}.
\end{itemize}

\begin{algorithm}[H]
    \caption{Simulator $\text{Sim}_{A^*}$}
    \begin{algorithmic}[1]
    \STATE  $A^*(\mathbf{\rho})$ outputs the message $\mathbf{\Psi}$.
    \STATE $\text{Sim}_{A^*}$ executes $2\lambda_{OT}$  sessions of ERE-Commitment to commit to dummy values by executing the equivocal simulator $\QQ_{\RE^*}$.
    \STATE $\text{Sim}_{A^*}$ obtains the challenged set $T \subset [2 \lambda_{OT}]$ output by $A^*$.
    
    \STATE $\text{Sim}_{A^*}$ samples $\hat{\boldsymbol{\theta}} \leftarrow \{+,\times\}^{\lambda_{OT}}$. $\text{Sim}_{A^*}$ measures the $i^{th}$ qubit of $\mathbf{\Psi}$, $\Psi_i$, for every $i \in T$ in the previous sampled bases, obtaining $\{\hat{x}\}_{i \in T}$. By applying Watrous' rewinding, $\text{Sim}_{A^*}$ opens the challenged bits to the correct values $\{\hat{x}_i, \hat{\theta}_i\}_{i \in T}$. 
    
    \STATE  $\text{Sim}_{A^*}$ discards the tested positions. $A$ sends $\boldsymbol{\theta}$ to $B$. $\text{Sim}_{A^*}$ obtains $\boldsymbol{\theta}$ and measures each of the qubits $\mathbf{\psi}_i \in \bar{T}$ in the bases $\theta_i$ to obtain $\hat{x}_i = x_i$.
    
    \STATE $\text{Sim}_{A^*}$ partitions the set $\bar{T}$ into two different subsets, by randomly sampling a bit $d_i \in \{0,1\}$ for each $i \in \bar{T}$, to create the sequences $I_0 = \{ i : d_i = 0\}$ and $I_{1} = \{ i : d_i = 1\}$. $\text{Sim}_{A^*}$ sends $(I_0, I_1)$ to $A^*$.
    
    \STATE $\text{Sim}_{A^*}$ obtains $(s_0, PRG(h(s_0,\boldsymbol{x}_0)) \oplus m_0, s_1, PRG(h(s_1,\boldsymbol{x}_1)) \oplus m_1)$, and computes $m_0$ and $m_1$. Both, $\boldsymbol{x}_0$ and $\boldsymbol{x}_1$ are the corrected measured strings $\hat{\boldsymbol{x}}_0$ and $\hat{\boldsymbol{x}}_1$ with the syndromes $Synd_0$ and $Synd_1$. If the error correcting protocol fails,  $\text{Sim}_{A^*}$ chooses $\boldsymbol{x}_0$ or/and $\boldsymbol{x}_1$ to correspondingly be the random strings $s_0, s_1 \leftarrow \{0,1\}^{\lambda_{OT}/2}$.

    \STATE If $A^*$ aborts at any point, $\text{Sim}_{A^*}$ sends $\perp$ to $\FUNC_{OT}(\cdot, b)$. Otherwise, it sends $(m_0,m_1)$ to $\FUNC_{OT}(\cdot, b)$. The output is the final state of $A^*$.

    \end{algorithmic}    \label{alg:OTsimA}
    \end{algorithm}

\begin{lemma} There exists a negligible function $\nu(\lambda)$, such that:

\begin{equation*}
    \begin{split}
        &\Big | \probP [\DD^*(\mathbf{\sigma}, \text{Hyb}_1) = 1 ] - \\
        &\probP[\DD^*(\mathbf{\sigma}, (\mathbf{\rho}_{\text{OUT},A^*}\langle A^*(\mathbf{\rho}), B(b) \rangle,\text{OUT}_B\langle A^*(\mathbf{\rho}), B(b) \rangle)) = 1] \Big| = \nu(\lambda).
    \end{split}
\end{equation*}
    
\end{lemma} \label{cl:eq}

\begin{proof}
 It follows from the equivocality property of ERE-Commitment.
\end{proof}

\begin{lemma}
$          \probP [\DD^*(\mathbf{\sigma}, \text{Hyb}_1) = 1 ] \equiv \probP [\DD^*(\mathbf{\sigma}, (\mathbf{\rho}_{SIM,\text{OUT},A^*},\text{OUT}_B)) = 1]$.
\end{lemma}

\begin{proof}
    It follows from the fact that the delay in the measurement of $\psi_i$ with $i \in T$ cannot be noticed by $A^*$ since, in both cases, the measurement bases are sampled at random, similar to \Cref{cl:eq2}. Moreover, in both cases the syndromes used for performing the error correction are the same ones. Therefore, if the syndromes sent by $A^*$ are not correct, both hybrids would lead to the same random distribution. Then, both extracted messages $m_0$ and $m_1$ are obtained by using the same error correction syndromes as in  $Hyb_1$, leading to the same distribution. 
\end{proof}

\subsection{Security against a malicious Bob}
In this section, we prove security against malicious Bob. We start by analyzing the number of qubits necessary in the protocol to achieve the desired level of security in \Cref{th:secmalrec}. Then, we prove the security of the protocol in \Cref{thm:sec-malicious-bob}.

\begin{lemma}[Distance bound for a malicious Bob in QOT]
Consider a QOT protocol between an honest Alice and a malicious Bob (\Cref{alg:OTProtocol}) with ERE-Commitment as the bit commitment subprotocol (\Cref{alg:EqExcomm}). Let $b \in \{0,1\}$ be a random bit and $K_b = h(r, \boldsymbol{x}_{I_b}) \in \{0,1\}^{\ell}$ and $K_{\bar{b}} = h(r, \boldsymbol{x}_{I_{\bar{b}}}) \in \{0,1\}^{\ell} $, the keys.  $h(r,x): \{0,1\}^m \times \{0,1\}^n \rightarrow \{0,1\}^{\ell} $ is a 2-universal hash function, and $I_0$  and $I_1$  are the two subsets in which $\bar{T} \subset [2 \lambda_{OT}]$ the unchallenged commitments are divided. Then, given $K_b$, for any sampling errors $\xi, \delta > 0$, a bit-flip error $\alpha \geq 0$ and a proportion $\vartheta \geq 0$ of leaked bits:

\begin{equation}
    \begin{split}
        & \Delta (\rho_{K_{\bar{b}}K_{b}E}, \frac{1}{2^{\ell}} \mathbbm{I} \otimes \rho_{K_{b}E})  \\
        &\leq \frac{1}{2} \cdot 2^{-\frac{1}{2}[\left(\frac{1}{2}-\xi - 2\vartheta\right)\frac{\lambda_{OT}}{2} -  h_2\left(\delta + \alpha + \chi\right)\frac{\lambda_{OT}}{2}\left(1 - 2\vartheta\right)-\ell-q]} \\
        & + \sqrt{6} \exp(-\lambda_{OT}\delta^2/100) + 2 \exp(- \xi^2\lambda_{OT}/2) ,
    \end{split}
    \label{eq:sec.boun.1}
\end{equation}
where $E$ denotes Bob's quantum state, $k$, the number of different seed families as described in Algorithm \ref{alg:EqExcomm} and $q$, the length of the syndrome needed for performing the error correction. $\mathbbm{I} \in \BB(\CC^{2^{\ell}})$ is the identity operator.
\label{th:secmalrec}
\end{lemma} 

\begin{proof}
    
The opening and checking process, which is based on sampling, follows a structure akin to the proof provided in \cite{BF12}. Interested readers can refer to it for a detailed analysis of the proof. This proof is done in the EPR setting in which Alice distributes half of an EPR pair to Bob in order for her to perform the sampling. Since Bob's and Alice's actions are equivalent in both settings, it suffices to prove the bound in the EPR setting. As the measurement bases are chosen randomly by Alice and Bob, the sampling strategy will be performed on a random subset $S \subseteq T$ defined as $S = \{ i \in T : \theta_i = \hat{\theta}_i \}$. Alice will accept the opening phase of the bit commitment if $w(\boldsymbol{x}|_S, \hat{\boldsymbol{x}}|_S) \leq \alpha$, where $w(\cdot) = d_H(\cdot)/n$, with $n \in \NN$ being the length of the bit-string. 

 Once the testing phase is accepted, the joint state $\ket{\psi_{A_{\bar{T}}E}}$ will be a superposition of states with relative Hamming weight $\delta$ close to $\alpha$ within $A_{\bar{T}}$ due to the statistical error of the sampling.  This error given by the sampling strategy will be bounded by Hoeffding's bound:
\begin{equation*}
    \epsilon^{\delta} \leq \sqrt{6} \exp(- \lambda_{OT} \delta^2/100).
\end{equation*}

This bound is obtained in Appendix B.4 of \cite{BF12}.  Assuming that Alice chooses her bases following a uniform random distribution and for $\boldsymbol{\theta}|_{\bar{T}}$ representing the subsets of bases that were not challenged, $d_H(\boldsymbol{\theta} |_{\bar{T}}, \hat{\boldsymbol{\theta}}|_{\bar{T}} ) \geq (1/2 - \gamma) |\bar{T}| = (1/2 - \gamma) \lambda_{OT}$ with Hoeffding's bound:

\begin{equation*}
   \epsilon^{\gamma} \leq 2 \exp(-2 \gamma^2\lambda_{OT}).
\end{equation*}

After the commitment and check subprotocols, the next step in the QOT protocol is when Alice sends the bases $\boldsymbol{\theta} |_{\bar{T}}$ to Bob, and Bob partitions $\bar{T}$ into two sets, $I_0$ and $I_1$. There exists a bit $b \in \{0,1\}$ such that $ d_H(\boldsymbol{\theta} |_{I_{\bar{b}}}, \hat{\boldsymbol{\theta}}|_{I_{\bar{b}}} ) \geq (1/2-2\gamma -2\vartheta)\lambda_{OT}/2$, regardless of how Bob splits the sets of indexes and taking the maximum advantage of the leaked bits, i.e., all the leaked bits are used in the set $I_{\bar{b}}$. Moreover, the real size of the non-sampled subset $I_{\bar{b}}$ for which the binary entropy in bounded is $\left(1 - 2\vartheta\right)$. By applying \Cref{eq:entropy} to our state $\rho_{X_{\bar{b}}A^{b}E}$, itself obtained from measuring the subsystem $A^{\bar{b}}$ on $\ket{\psi_{A^{\bar{b}}A^{b}E}}$, privacy amplification gives, for $\xi = 2\gamma$:
\begin{equation*}
\begin{split}
H_{\text{min}}(X_{\bar{b}}|A^{b}E) \geq & d_H(\boldsymbol{\theta}|_{I_{\bar{b}}}, \hat{\boldsymbol{\theta}}|_{I_{\bar{b}}} ) - h_2(\delta + \alpha +\eta + \zeta)|I_{\bar{b}}| \\
& \geq \left(\frac{1}{2}-\xi - 2\vartheta\right)\frac{\lambda_{OT}}{2} -  h_2\left(\delta + \alpha + \chi\right)\frac{\lambda_{OT}}{2}\left(1 - 2\vartheta\right).
    \end{split} 
\end{equation*}

Given a $\chi$-relaxed extractable ERE-Commitment, with $\chi = \eta + \zeta = \omega(\log^2(k)/k)$, ERE-Commitment is relaxed extractable. 

\end{proof}

We now move on to prove the security of the protocol against a quantum-poly time malicious Bob and an honest Alice.  The sequence of hybrid arguments for this case is:

\begin{itemize}
    \item \textbf{Hybrid 0}. The $Real$ distribution given by Algorithm \ref{alg:OTProtocol}.
    \item \textbf{Hybrid 1}. Similar to $\text{Hyb}_0$ except in Step 2, in which we execute the extractable simulator $\QQ_{\RE^*,com}$ of Algorithm \ref{alg:extractor}.
    \item \textbf{Hybrid 2}. The $Ideal$ distribution given by Algorithm \ref{alg:OTsimB}.
\end{itemize}

\begin{algorithm}[H]
    \caption{Simulator $\text{Sim}_{B^*}$}
    \begin{algorithmic}[1]
    \STATE  $\text{Sim}_{B^*}$ chooses $\boldsymbol{x} \leftarrow \{0,1\}^{2\lambda_{OT}}$ and $\boldsymbol{\theta} \leftarrow \{+,\times\}^{2\lambda_{OT}}$ and sends $\ket{\boldsymbol{x}}_{\boldsymbol{\theta}}$ to $B^*$.
    \STATE $B^*$ and $\text{Sim}_{B^*}$ execute $2\lambda_{OT}$  sessions of ERE-Commitment. Then, $\text{Sim}_{B^*}$ extracts the committed values $\left( \hat{\theta}_i, \hat{x}_i\right)_{i \in [2\lambda_{OT}]}$ by applying the extractable simulator given by Algorithm \ref{alg:extractor} and aborts if the number of extracted commitments equal to $\mathbbold{\qm}$ are larger than $\chi = \eta + \zeta$.
    
    \STATE $\text{Sim}_{B^*}$ challenges a random subset $T \subset 2 \lambda_{OT}$, such that $\abs{T} = \lambda_{OT}$.
    
    \STATE $\text{Sim}_{B^*}$ and $B^*$ execute ERE-Decommitment for every $i \in T$. Then, $\text{Sim}_{B^*}$ obtains the set $\left( (\hat{x}_i, \hat{\theta}_i)\right)_{i \in T}$ and checks that $x_i = \hat{x}_i$ whenever $\theta_i = \hat{\theta}_i$, up to an experimental error $\alpha$. If the test does not pass, $\text{Sim}_{B^*}$ aborts.
    
    \STATE  $\text{Sim}_{B^*}$ discards the tested positions and sends $\boldsymbol{\theta}$ to $B^*$. 
    
    \STATE $\text{Sim}_{B^*}$ obtains $I_0$ and $I_1$. Let $S$ be the sequence of indices in $I_0$ such that $\theta_i \neq \hat{\theta}_i$. If $\abs{S} \geq \left(1/2-\xi \right) \lambda_{OT}/2$ set $b=1$, otherwise set $b = 0$. 
    
    \STATE $\text{Sim}_{B^*}$ obtains $m_b$ from $\FUNC_{OT}$ and sets $m_{\bar{b}} = 0$.

    \STATE  $\text{Sim}_{B^*}$ samples seeds $s_0,s_1 \leftarrow \{0,1\}^{k(\lambda_{PQS})}$ and sends $(s_0, PRG(h(s_0,\boldsymbol{x}_0)) \oplus m_0, s_1, PRG(h(s_1,\boldsymbol{x}_1)) \oplus m_1)$, where $\boldsymbol{x}_0 = \{x_i : i \in I_0\}$ and $\boldsymbol{x}_1 = \{x_i : i \in I_1\}$. He also sends the syndromes $Synd_0$ and $Synd_1$ for correcting a proportion $\alpha$ of errors of the bit strings $\boldsymbol{x}_0$ and $\boldsymbol{x}_1$. $\text{Sim}_{B^*}$ outputs the final state of $B^*$.

    \end{algorithmic}    \label{alg:OTsimB}
    \end{algorithm}

\begin{theorem} 
\label{thm:sec-malicious-bob}
There exists a negligible function $\nu(\lambda)$, such that,

\begin{align*}
        & \Big | \probP [\DD^*(\mathbf{\sigma}, \text{Hyb}_1) = 1 ] \\
        &- \probP[\DD^*(\mathbf{\sigma}, (\text{OUT}_A\langle A(m_0,m_1), B^*(b) \rangle),\mathbf{\rho}_{\text{OUT},B^*}\langle A(m_0,m_1), B^*(b) \rangle ) = 1] \Big| = \nu(\lambda)
\end{align*}
    
\end{theorem} \label{cl:eq3}

\begin{proof}
    It follows from the $\omega(\log^2(k)/k)$-relaxed-extractability definition of the bit commitment.     
\end{proof}

\begin{lemma}
$          \left | \probP [\DD^*(\mathbf{\sigma}, \text{Hyb}_1) = 1 ] - \probP [\DD^*(\mathbf{\sigma}, (\text{OUT}_A,\mathbf{\rho}_{SIM,\text{OUT},B^*} ) = 1] \right| = \nu(\lambda)$.
\end{lemma}

\begin{proof}
    This can be proven by contradiction. Suppose that there exists a distinguisher such that:
    \begin{equation*}
        \left | \probP [\DD^*(\mathbf{\sigma}, \text{Hyb}_1) = 1 ] - \probP [\DD^*(\mathbf{\sigma}, (\text{OUT}_A,\mathbf{\rho}_{SIM,\text{OUT},B^*} ) = 1] \right| \geq \frac{1}{poly(\lambda)}.
    \end{equation*}
    
This would mean that there is a non-uniform quantum poly-time distinguisher $\DD^*$ that can distinguish between the messages $m_{\bar{b}}$ obtained in $\text{Hyb}_1$ and $m_{\bar{b}}$ obtained in $\text{Hyb}_2$. Since $m_{\bar{b}}$ is encoded as $PRG(h(s_{\bar{b}},\hat{\boldsymbol{x}}_{\bar{b}})) \oplus m_{\bar{b}}$, this would be equivalent to distinguishing between the seeds $h(s_{\bar{b}},\hat{\boldsymbol{x}}_{\bar{b}})$ of $\text{Hyb}_1$ and $\text{Hyb}_2$. This leads to a contradiction since $h(s_{\bar{b}},\hat{\boldsymbol{x}}_{\bar{b}})$ is statistically close to a random string due to privacy amplification, as proven in \Cref{th:secmalrec}. Moreover, in both distributions, the applied syndromes for the error correction of the strings $\hat{\boldsymbol{x}}_{\bar{b}}$ and $\hat{\boldsymbol{x}}_{b}$ are the same ones, leading to the same corrected strings.

\end{proof}

\begin{corollary} Given one implementation of the QOT protocol described in Algorithm \ref{alg:OTProtocol} with the quantum distributed key $\boldsymbol{x} \in \{0,1\}^{2\lambda_{OT}}$, a quantity $n_{OT}$  of 1-out-of-2 oblivious transfers can be distilled, where $n_{OT}=\lfloor\lambda_{OT}/v\rfloor$ and where $v$ is the size of the preimage of a hash function such that $h(r,x): \{0,1\}^{m(\lambda_{PQS})} \times \{0,1\}^m \rightarrow \{0,1\}^{\lambda_{PQS}}$. The parameter $v$ is determined in such a way that:
\begin{equation*}
    \begin{split}
        & \Delta (\rho_{K_{\bar{b}}K_{b}E}, \frac{1}{2^{\ell}} \mathbbm{I} \otimes \rho_{K_{b}E})  \\
        &\leq \frac{1}{2} \cdot 2^{-\frac{1}{2}[(\frac{1}{2}-\xi)\frac{v}{2} - \vartheta \lambda_{OT} -  h_2(\delta + \alpha + \chi)(\frac{v}{2} - \vartheta \lambda_{OT})-\ell - q]} \\
        & + \sqrt{6} \exp(-\lambda_{OT}\delta^2/100) + 2 \exp(- \xi^2\lambda_{OT}/2) ,
    \end{split}
\end{equation*}
with $\ell=\lambda_{PQS}$. 
 
 \label{col:distillation}
\end{corollary}

\begin{proof}

The proof is similar to the one in \Cref{th:secmalrec}. The sampling requires a large number of $\hat{x}_i$ and $\hat{\theta}_i$. On step 7 of \Cref{alg:OTProtocol}, Bob can choose subsets of length $v$ of $\hat{\theta}_i$ and $\hat{x}_i$ to generate $n_{OT}$ seeds $s_i \in \{0,1\}^{\lambda_{PQS}}$.  By selecting a large enough $v$, we avoid finite-set sampling errors, leading $\xi$ and $\delta$ to be representative of the whole population as in \Cref{th:firsteq}, and thus:

\begin{equation}
 H_{\text{min}}(X_{\bar{b}}|A^{b}E) \geq \left(\frac{1}{2}-\xi\right)\frac{v}{2} - \vartheta \lambda_{OT} -  h_2(\delta + \alpha + \chi)\left(\frac{v}{2} - \vartheta \lambda_{OT}\right).
\end{equation}

In order to obtain $n_{OT}$ 1-out-of-2 OTs, the first difference with respect to Algorithm \ref{alg:OTProtocol} is that, in Step 7, Bob chooses the tuple of sets $(I_{b,j})_{j \in [n_{OT}]}$ and $(I_{\bar{b},j})_{j \in [n_{OT}]}$,  such that $I_{b,j} = \{i : \theta_i = \hat{\theta}_i\} $, $I_{b,j} \cap I_{b,j'} = \varnothing$,  $I_{\bar{b},j} = \{i : \theta_i \neq \hat{\theta}_i\} $ and $I_{\bar{b},j} \cap I_{\bar{b},j'} = \varnothing$. Then Bob sends $((I_0,I_1)_j)_{j \in [n_{OT}]}$ to Alice.  The second difference is that, in Step 8, we define  $\mathbf{x}_{0,j} = \{x_i : i \in I_{0,j}\}$ and $\mathbf{x}_{1,j} = \{x_i : i \in I_{1,j}\}$.  Alice can then send $((m_0,m_1)_j)_{j \in [n_{OT}]}$ different encoded messages. For each encoded message, Alice chooses at random a pair $(I_0,I_1)_{j}$ over the entire tuple $((I_0,I_1))_{j}\}_{j \in [n_{OT}]}$ and shares it with Bob, who can decrypt $m$ messages $m_b$ by using the strings $(\hat{\mathbf{x}}_{b,j})_{j \in [n_{OT}]}$ with $\hat{\mathbf{x}}_{b,j} = \{ \hat{x}_i : i \in I_{b,j} \}$. 
\end{proof}

\section*{Acknowledgements}

AY acknowledges Lucas Hanouz for the discussions related with error correction. We also want to acknowledge Matías R. Bolaños for highlighting the necessity of adding multiphoton events and suggesting some improvements in the protocol. Lastly, we want to acknowledge anonymous reviewers for their feedback. ED, ABG and AY are supported by the European Union's Horizon Europe Framework Programme under the Marie Skłodowska Curie Grant No. 101072637, Project Quantum-Safe Internet (QSI). ED and PL acknowledge PEPR integrated project QCommTestbed ANR-22-PETQ-0011 part of Plan France 2030, VY acknowledges the support of European Research Council Starting Grant QUSCO, 758911. ED and AI acknowledge the support of QuantERA QuantaGENOMICS under Grant Agreement No. 101017733. ABG is supported by ANR JCJC TCS-NISQ ANR-22-CE47-0004.
This work was done in part while ABG was visiting the Simons Institute for the Theory of Computing.

\bibliography{bibliography_quantum.bib}

\begin{thebibliography}{10}

\bibitem{BB84}
C.~H. Bennett and G.~Brassard.
\newblock ``Quantum cryptography: Public key distribution and coin tossing''.
\newblock \href{https://dx.doi.org/10.1109/CSSP.1984.4298688}{Proceedings of IEEE International Conference on Computers, Systems and Signal ProcessingPages 175--179}~(1984).

\bibitem{LoChau97}
Hoi-Kwong Lo and Hoi~Fung Chau.
\newblock ``Is quantum bit commitment really possible?''.
\newblock \href{https://dx.doi.org/10.1103/PhysRevLett.78.3410}{Physical Review Letters {\bf 78}, 3410--3413}~(1997).

\bibitem{Mayers97}
Dominic Mayers.
\newblock ``Unconditionally secure quantum bit commitment is impossible''.
\newblock \href{https://dx.doi.org/10.1103/PhysRevLett.78.3414}{Physical Review Letters {\bf 78}, 3414--3417}~(1997).

\bibitem{DFL+09}
Ivan Damgård, Serge Fehr, Carolin Lunemann, Louis Salvail, and Christian Schaffner.
\newblock ``Improving the security of quantum protocols via commit-and-open''.
\newblock In Shai Halevi, editor, CRYPTO 2009.
\newblock \href{https://dx.doi.org/10.1007/978-3-642-03356-8_24}{Volume 5677 of LNCS, page 408–427}.
\newblock Heidelberg~(2009). Springer.

\bibitem{IR90}
Russell Impagliazzo and Steven Rudich.
\newblock ``Limits on the provable consequences of one-way permutations''.
\newblock In Shafi Goldwasser, editor, Advances in Cryptology --- CRYPTO'88.
\newblock \href{https://dx.doi.org/10.1007/0-387-34799-2_2}{Volume 403 of Lecture Notes in Computer Science, pages 8--26}.
\newblock New York, NY~(1990). Springer.

\bibitem{GLSV21}
Alex~B. Grilo, Huijia Lin, Fang Song, and Vinod Vaikuntanathan.
\newblock ``Oblivious transfer is in {MiniQCrypt}''.
\newblock In Anne Canteaut and Fran{\c{c}}ois-Xavier Standaert, editors, Advances in Cryptology -- EUROCRYPT 2021.
\newblock \href{https://dx.doi.org/10.1007/978-3-030-77886-6_18}{Volume 12697 of Lecture Notes in Computer Science, pages 531--561}.
\newblock Springer~(2021).

\bibitem{BCKM21}
James Bartusek, Andrea Coladangelo, Dakshita Khurana, and Fermi Ma.
\newblock ``One-way functions imply secure computation in a quantum world''.
\newblock In Tal Malkin and Chris Peikert, editors, Advances in Cryptology -- CRYPTO 2021.
\newblock \href{https://dx.doi.org/10.1007/978-3-030-84242-0_17}{Volume 12825 of Lecture Notes in Computer Science, pages 467--496}.
\newblock Springer~(2021).

\bibitem{IPS08}
Yuval Ishai, Manoj Prabhakaran, and Amit Sahai.
\newblock ``Founding cryptography on oblivious transfer -- efficiently''.
\newblock In David Wagner, editor, Advances in Cryptology -- CRYPTO 2008.
\newblock \href{https://dx.doi.org/10.1007/978-3-540-85174-5_32}{Volume 5157 of Lecture Notes in Computer Science, pages 572--591}.
\newblock Springer~(2008).

\bibitem{Kil88}
Joe Kilian.
\newblock ``Founding cryptography on oblivious transfer''.
\newblock In Proceedings of the Twentieth Annual ACM Symposium on Theory of Computing.
\newblock \href{https://dx.doi.org/10.1145/62212.62215}{Pages 20--31}.
\newblock STOC '88. ACM~(1988).

\bibitem{MY22}
Tomoyuki Morimae and Takashi Yamakawa.
\newblock ``Quantum commitments and signatures without one-way functions''.
\newblock \href{https://dx.doi.org/10.1007/978-3-031-15802-5_10}{Page 269–295}.
\newblock Springer Nature Switzerland. ~(2022).

\bibitem{AQY22}
Prabhanjan Ananth, Luowen Qian, and Henry Yuen.
\newblock ``Cryptography from pseudorandom quantum states''.
\newblock In Yevgeniy Dodis and Thomas Shrimpton, editors, Advances in Cryptology -- CRYPTO 2022.
\newblock \href{https://dx.doi.org/10.1007/978-3-031-15802-5_8}{Volume 13507 of Lecture Notes in Computer Science, pages 208--236}.
\newblock Cham~(2022). Springer.

\bibitem{ACA23}
Costantino Agnesi, Marco Avesani, Federico Berra, Luca Calderaro, Sebastiano Cocchi, Giulio Foletto, Davide Scalcon, Andrea Stanco, Giuseppe Vallone, and Paolo Villoresi.
\newblock ``A versatile and low-error polarization encoder for quantum communications''.
\newblock In Optica Quantum 2.0 Conference and Exhibition.
\newblock \href{https://dx.doi.org/10.1364/QUANTUM.2023.QTu3A.25}{Page QTu3A.25}.
\newblock Technical Digest Series. Optica Publishing Group~(2023).

\bibitem{GBR23}
Fadri Gr{\"u}nenfelder, Alberto Boaron, Giovanni~V. Resta, Matthieu Perrenoud, Davide Rusca, Claudio Barreiro, Rapha{\"e}l Houlmann, Rebecka Sax, Lorenzo Stasi, Sylvain El-Khoury, Esther H{\"a}nggi, Nico Bosshard, F{\'e}lix Bussi{\`e}res, and Hugo Zbinden.
\newblock ``Fast single-photon detectors and real-time key distillation enable high secret-key-rate quantum key distribution systems''.
\newblock \href{https://dx.doi.org/10.1038/s41566-023-01168-2}{Nature Photonics {\bf 17}, 422--426}~(2023).

\bibitem{PSG23}
Yoann Pelet, Gr{\'e}gory Sauder, Mathis Cohen, Laurent Labont{\'e}, Olivier Alibart, Anthony Martin, and S{\'e}bastien Tanzilli.
\newblock ``Operational entanglement-based quantum key distribution over 50 km of field-deployed optical fibers''.
\newblock \href{https://dx.doi.org/10.1103/PhysRevApplied.20.044006}{Physical Review Applied {\bf 20}, 044006}~(2023).

\bibitem{ZMM24}
Mujtaba Zahidy, Mikkel~T. Mikkelsen, Ronny M{\"u}ller, Beatrice Da~Lio, Martin Krehbiel, Ying Wang, Nikolai Bart, Andreas~D. Wieck, Arne Ludwig, Michael Galili, S{\o}ren Forchhammer, Peter Lodahl, Leif~K. Oxenl{\o}we, Davide Bacco, and Leonardo Midolo.
\newblock ``Quantum key distribution using deterministic single-photon sources over a field-installed fibre link''.
\newblock \href{https://dx.doi.org/10.1038/s41534-023-00800-x}{npj Quantum Information {\bf 10}, 2}~(2024).

\bibitem{CK88}
Claude Crépeau and Joe Kilian.
\newblock ``Weakening security assumptions and oblivious transfer (abstract)''.
\newblock In Advances in Cryptology - CRYPTO '88, 8th Annual International Cryptology Conference, Santa Barbara, California, USA, August 21-25, 1988, Proceedings.
\newblock \href{https://dx.doi.org/10.1007/0-387-34799-2_1}{Volume 403 of Lecture Notes in Computer Science, pages 2--7}.
\newblock Springer~(1988).

\bibitem{BBCS92}
Charles~H. Bennett, Gilles Brassard, Claude Crépeau, and Maria~H. Skubiszewska.
\newblock ``Practical quantum oblivious transfer''.
\newblock In Joan Feigenbaum, editor, Advances in Cryptology — {CRYPTO} '91.
\newblock \href{https://dx.doi.org/10.1007/3-540-46766-1_29}{Volume 576 of Lecture Notes in Computer Science, pages 351--366}.
\newblock Berlin, Heidelberg~(1992). Springer.

\bibitem{BF12}
Niek~J Bouman and Serge Fehr.
\newblock ``Sampling in a quantum population, and applications''.
\newblock In Annual Cryptology Conference.
\newblock \href{https://dx.doi.org/10.1007/978-3-642-14623-7_39}{Pages 724--741}.
\newblock Springer~(2010).

\bibitem{Wat05}
John Watrous.
\newblock ``Zero-knowledge against quantum attacks''.
\newblock In Proceedings of the Thirty-Eighth Annual ACM Symposium on Theory of Computing.
\newblock \href{https://dx.doi.org/10.1145/1132516.1132560}{Pages 296--305}.
\newblock STOC '06. ACM~(2006).

\bibitem{Ren05}
Renato Renner.
\newblock ``Security of quantum key distribution''.
\newblock \href{https://dx.doi.org/10.1142/S0219749908003256}{International Journal of Quantum Information {\bf 6}, 1--127}~(2008).

\bibitem{ABKK23}
Amit Agarwal, James Bartusek, Dakshita Khurana, and Nishant Kumar.
\newblock ``A new framework for quantum oblivious transfer''.
\newblock In Annual International Conference on the Theory and Applications of Cryptographic Techniques.
\newblock \href{https://dx.doi.org/10.1007/978-3-031-30545-0_13}{Pages 363--394}.
\newblock Springer~(2023).

\bibitem{CCLY23}
Nai-Hui Chia, Kai-Min Chung, Xiao Liang, and Takashi Yamakawa.
\newblock ``Post-quantum simulatable extraction with minimal assumptions: Black-box and constant-round''.
\newblock In Annual International Cryptology Conference.
\newblock \href{https://dx.doi.org/10.1007/978-3-031-15982-4_18}{Pages 533--563}.
\newblock Springer~(2022).

\bibitem{DFSS06}
Ivan~B Damg{\aa}rd, Serge Fehr, Louis Salvail, and Christian Schaffner.
\newblock ``Cryptography in the bounded-quantum-storage model''.
\newblock \href{https://dx.doi.org/10.1137/060651343}{SIAM Journal on Computing {\bf 37}, 1865--1890}~(2008).

\bibitem{EHG+13}
C.~Erven, N.~Ng, N.~Gigov, R.~Laflamme, S.~Wehner, and G.~Weihs.
\newblock ``An experimental implementation of oblivious transfer in the noisy storage model''.
\newblock \href{https://dx.doi.org/10.1038/ncomms4418}{Nature Communications{\bf 5}}~(2014).

\bibitem{LPAK23}
Cosmo Lupo, James~T. Peat, Erika Andersson, and Pieter Kok.
\newblock ``Error-tolerant oblivious transfer in the noisy-storage model''.
\newblock \href{https://dx.doi.org/10.1103/physrevresearch.5.033163}{Physical Review Research{\bf 5}}~(2023).

\bibitem{Ga60}
R.~G. Gallager.
\newblock ``Low density parity check codes''.
\newblock Sc.d. thesis.
\newblock Massachusetts Institute of Technology.
\newblock Cambridge, MA~(1960).

\bibitem{Naor91}
Moni Naor.
\newblock ``Bit commitment using pseudorandomness''.
\newblock \href{https://dx.doi.org/10.1007/BF00196774}{Journal of Cryptology {\bf 4}, 151--158}~(1991).

\bibitem{HLS24}
Yassine Hamoudi, Qipeng Liu, and Makrand Sinha.
\newblock ``The nisq complexity of collision finding''.
\newblock In Annual International Conference on the Theory and Applications of Cryptographic Techniques.
\newblock \href{https://dx.doi.org/doi.org/10.1007/978-3-031-58737-5_1}{Pages 3--32}.
\newblock Springer~(2024).

\bibitem{NIST}
{National Institute of Standards and Technology}.
\newblock ``Submission requirements and evaluation criteria for the post-quantum cryptography standardization process''.
\newblock \url{https://csrc.nist.gov/CSRC/media/Projects/Post-Quantum-Cryptography/documents/call-for-proposals-final-dec-2016.pdf}~(2016).
\newblock Accessed: 2025-11-20.

\end{thebibliography}

\appendix
\section{Distance bound for a malicious Alice}
\label{appendix}

Since Naor's commitment is computationally hiding, and so is the extractable layer, the security against a malicious Alice  will be bound by the preservation of the hiding property, as discussed in \Cref{col:EqEreCom}. This bound is given by the $\Delta$-distance between the key generated by ERE-Commitment and a string of uniformly random bits of length $\lambda_{PQS}$.

\begin{theorem}[Distance bound in ERE-Commitment for a malicious receiver]
Consider ERE-Commitment (see Algorithm \ref{alg:EqExcomm}), between an honest sender and a malicious receiver with a computationally hiding and ($\eta$, $2^{-\eta n}$)-relaxed statistically binding EqCommitment as base commitment. Let $K \coloneqq g(r, \boldsymbol{x})$ be the key in $\{0,1\}^{\ell}$ output by Alice, where $h(r,x): \{0,1\}^m \times \{0,1\}^n \rightarrow \{0,1\}^{\ell} $ is a universal hash function, and $\bar{T} \subset [4\lambda_{EX}]$, the set of commitments that have not been challenged. Moreover, $m \in \bar{T}$ is a randomly chosen subset of non-challenged commitments that are input in the hash function.  Then, for any sampling error $\xi, \delta > 0$, a bit-flip error $\alpha \geq 0$ and a proportion $\vartheta \geq 0$ of leaked bits:
\begin{equation}
    \begin{split}
        & \Delta (\rho_{K E}, \frac{1}{2^{\ell}} \mathbbm{I} \otimes \rho_{E})  \leq \\
        &\leq \frac{1}{2} \cdot 2^{-\frac{1}{2}[\left(\frac{1}{2}-\xi\right)m - 2\vartheta \lambda_{EX}  -  h_2\left(\delta + \alpha +\eta \right) \left(m  - 2\vartheta \lambda_{EX}\right) -\ell - q]} \\
        & + \sqrt{6} \exp(-2\lambda_{EX}\delta^2/100) + 2 \exp(-4 \xi^2 \lambda_{EX}) ,
    \end{split} \label{eq:recseq}
\end{equation}
where $E$ denotes Bob's quantum state, $q$ the length of the syndrome needed for performing the error correction code and $\mathbbm{I} \in \BB(\CC^{2^{\ell}})$ is the identity operator.

\label{th:secmalsen}
\end{theorem} 

\begin{proof}
    Since the sampling strategy is the same as the one of the OT layer, the sampling errors will be the same.  The main difference is that, for this layer, the set of not challenged commitments $\bar{T} \subset [2\lambda_{EX}]$ is not divided into two different subsets. Therefore, the malicious Bob has less prior information:

\begin{equation}
\begin{split}
       H_{\text{min}}(X|E) & \geq d_H(\boldsymbol{\theta} |_{\bar{T}}, \hat{\boldsymbol{\theta}}|_{\bar{T}}) - h_2(\delta + \alpha+\eta )|{\bar{T}}| \\
       & \geq \left(\frac{1}{2}-\xi -\vartheta \right)2\lambda_{EX} -  h_2\left(\delta + \alpha +\eta \right)2\lambda_{EX}(1-\vartheta). 
\end{split}
\end{equation}
As the subset $m \in {\bar{T}}$ is randomly chosen from the set ${\bar{T}}$, the sampling error for the bases and the states will be the same ones as for the set ${\bar{T}}$ for a large enough $m$. Therefore, its min-entropy will be:

\begin{equation}
    H_{\text{min}}(X_m|E) \geq \left(\frac{1}{2}-\xi\right)m - 2\vartheta \lambda_{EX}  -  h_2\left(\delta + \alpha +\eta \right) \left(m  - 2\vartheta \lambda_{EX}\right). 
\end{equation}

Given a $\eta$-relaxed statistical binding EqCommitment, the probability of being relaxed statistical binding is $1-negl(k)$ when $\eta =\omega(\log^2(k)/k)$

\end{proof}

\section{Benchmark of QOT protocols}\label{app:benchmark}
The numerical benchmark presented in \Cref{table:benchmark} is obtained by numerically optimizing the concrete security bounds of the corresponding protocols.

While \cite{BCKM21} does not provide explicit analytical bounds, in \cite[Sec.~5.2 and Sec.~6.4]{BCKM21} the authors refer to \cite{BF12} for the statistical binding proof. Hence, arguing as in \Cref{th:secmalrec}, one can derive the following bound for the OT layer:
\begin{equation*}
    \begin{split}
        & \Delta (\rho_{K_{\bar{b}}K_{b}E}, \frac{1}{2^l} \id \otimes \rho_{K_{b}E})  \leq \\
        &\leq \frac{1}{2} \cdot 2^{-\frac{1}{2}[(1-\xi -  h(\delta))4\lambda_{OT}-l]} + \sqrt{6} \exp(-\delta^2 8\lambda_{OT}/100) + 2 \exp(- \xi^24\lambda_{OT}) ,
    \end{split}
\end{equation*}
where $\xi,\delta > 0$ denote sampling parameters and $h(r,x): \{0,1\}^m \times \{0,1\}^n \rightarrow \{0,1\}^{\ell}$ is a universal hash function. Similarly, following the same arguments as in \Cref{th:secmalsen}, one obtains the following security bound for the extractable layer:

\begin{equation*}
    \begin{split}
        & \Delta (\rho_{K_{\bar{b}}K_{b}E}, \frac{1}{2^l} \id \otimes \rho_{K_{b}E})  \leq \\
        &\leq \frac{1}{2} \cdot 2^{-\frac{1}{2}[(\frac{1}{2}-\xi -  h(\delta))\lambda^2_{EX}-1]} + \sqrt{6} \exp(-\lambda^3_{EX}\delta^2/100) + 2 \exp(-2 \xi^2\lambda^3_{EX}),
    \end{split}
\end{equation*}
where the security parameters are taken following \cite[Protocol~3]{BCKM21}. Following the structure of the protocol in \cite{BCKM21}, the total number of BB84 states is given by
\begin{align*}
    \lambda_{\mathrm{BCKM21}} =  16 \cdot \lambda_{OT}\cdot \lambda_{eq}  \cdot 4  \cdot 2\cdot \lambda^3_{EX} + 16\cdot \lambda_{OT},
\end{align*}
since each OT layer consists of $\lambda_{eq}$ equivocal commitments, each of them compounded by four executions of the extractable commitment. Therefore, by setting $\Delta\leq \varepsilon$ with $\varepsilon = 10^{-15}$ and $\lambda_{eq} = 128$, we obtain the value $\lambda_{\mathrm{BCKM21}} = 2.27\times 10^{13}$ reported in \Cref{table:benchmark}.

We now turn to the protocol proposed in \cite{ABKK23}. Among their constructions, we focus on the ones based on BB84 states, namely \cite[Protocol~9]{ABKK23} and \cite[Protocol~10]{ABKK23}, which implement string OT with random bases. The bound for the 3-round communication protocol of \cite[Protocol~9]{ABKK23} is given in \cite[Theorem~B.1]{ABKK23}. It requires a total of $\lambda_{\mathrm{3OT \, ABKK22}}= 23000 \lambda$ BB84 states, where $\lambda$ is determined by the following security bounds:
\begin{align*}
    \mu_{\mathrm{R}^*} & =\frac{\sqrt{5}}{2^\lambda}+\frac{4 q_{RO}}{2^{18 \lambda}}+\frac{148(q_{RO}+46000 \lambda+1)^3+1}{2^{2 \lambda}}+\frac{368000 q_{RO} \lambda}{2^\lambda}, \\
    \mu_{\mathrm{S}^*} &=\frac{430 q_{RO} \sqrt{\lambda}}{2^\lambda} \,,
\end{align*}
where $\mu_{\mathrm{R}^*}$ and $\mu_{\mathrm{S}^*}$ are, respectively, the security bounds against a malicious receiver and a malicious sender, and $q_{RO}$ is an upper bound on the total number of random-oracle queries made by the adversary.

Similarly, \cite[Protocol~10]{ABKK23}, which yields a 4-round OT protocol, requires a total of $\lambda_{\mathrm{4OT \, ABKK22}}=10300 \lambda$ BB84 states, where $\lambda$ is determined by the bounds in \cite[Theorem~B.2]{ABKK23}:
\begin{align*}
    \mu_{\mathrm{R}^*} & =\frac{\sqrt{5}}{2^\lambda}+\frac{1}{2^{9 \lambda}}+\frac{148(q_{RO}+2 n+1)^3+1}{2^{2 \lambda}}+\frac{16 q_{RO} n}{2^\lambda}, \\
    \mu_{\mathrm{S}^*} &=\frac{288 q_{RO} \sqrt{\lambda}}{2^\lambda} \,,
\end{align*}
where, following \cite[Appendix~B.2]{ABKK23}, we take $n = 10300\,\lambda$.

As discussed in \cite{ABKK23}, the random oracle (RO) can be instantiated with SHA-256. Following \cite[Theorem~5.1]{HLS24}, the optimal success probability for solving the preimage search problem on an $m$-bit random oracle with $q_{RO}$ quantum queries is $\Theta(q_{RO}^2/2^m)$. Taking $m = 256$ and imposing $p_{\mathrm{succ}}\leq 2^{-128}$ yields $q_{RO} \leq 2^{64}$, corresponding to 128-bit post-quantum preimage security. This upper bound on the number of oracle queries is also consistent with the NIST recommendation of considering at most $2^{64}$ oracle-like queries in its post-quantum security definitions \cite[Sec.~4.A.2/4.A.4]{NIST}.

By imposing $\mu_{\mathrm{R}^*} = \mu_{\mathrm{S}^*} \leq \varepsilon$ with $\varepsilon = 10^{-15}$, and using $q_{RO}= 2^{64}$, we obtain the numerical values $\lambda_{\mathrm{3OT \, ABKK22}} = 3.22 \times 10^6$  and $\lambda_{\mathrm{4OT \, ABKK22}}= 1.43\times 10^6$ BB84 states reported in \Cref{table:benchmark}.

For our protocol, the optimization is performed by targeting at most $\varepsilon = 10^{-15}$ and fixing $\alpha = 0.006$ and $\vartheta= 0.001$ in \Cref{eq:sec.boun.1} and \Cref{eq:recseq}, with a error correction syndrome of $q = 0.1\lambda$ \footnote{In the numerical optimization, we assume that the experimental error is uniformly distributed.}. It is important to emphasize that this is the only benchmark that explicitly incorporates transmission errors, since the protocols in \cite{BCKM21,ABKK23} are not robust to such errors. For completeness, we also include in \Cref{table:benchmark} the benchmark corresponding to a noiseless implementation of our protocol.

\end{document}